%% file: paper_top.tex
\documentclass[acmsmall,10pt]{acmart}

\renewcommand\footnotetextcopyrightpermission[1]{} 
\usepackage{graphics,xspace,amssymb,listings}
\usepackage{alltt}
\usepackage{stmaryrd}
\usepackage{xcolor}
\usepackage{listings}
\AtBeginDocument{%
  \providecommand\BibTeX{{%
    \normalfont B\kern-0.5em{\scshape i\kern-0.25em b}\kern-0.8em\TeX}}}

 \newcommand{\xnearrow}[1]{%
  {\left\nearrow\vbox to #1{}\right.\kern-\nulldelimiterspace}
}

\usepackage{amssymb}
\usepackage{graphicx}
\usepackage{amsmath}
\usepackage{listings}
\usepackage{amsfonts}
\usepackage{algpseudocode}
\usepackage{graphics,xspace,amssymb,listings}
\usepackage{color}
\usepackage{microtype}
\usepackage{enumitem}
\usepackage{fancybox}
\usepackage{soul}
\usepackage{url}
\usepackage{multicol}
\usepackage{mdwlist}
\usepackage{epstopdf}
\usepackage{cancel}
\usepackage{tikz-cd}
\usepackage{comment}
\usepackage{stmaryrd}
\usepackage{natbib} 
\setcitestyle{authoryear,sort&compress}
\algrenewcommand\alglinenumber[1]{\scriptsize #1:}

\lstnewenvironment{javan}
{\lstset{language=java,
  basicstyle=\small\sffamily,
  tabsize=2,
  xleftmargin=16pt,
  numbers=left, 
  numberstyle=\tiny,
  columns=fullflexible,
  captionpos=b,
  escapechar=\%}}  {}

\newenvironment{semantics}{
  \begin{displaymath}}{
  \end{displaymath}
}
\newcommand{\ntyperule}[3]{ 
  \begin{array}{c} 
    \textsc{\scriptsize ({#1})} \\ 
    #2 \\[1mm]
    \hline 
    \raisebox{0pt}[12pt][0pt]{\ensuremath{#3}}
  \end{array}}

\newcommand\secref[1]{Sect.~\ref{#1}}
\newcommand\secsref[1]{Sects.~\ref{#1}}
\newcommand\figref[1]{Fig.~\ref{#1}}

\newcommand\defref[1]{Def.~\ref{#1}}

\newcommand{\todo}[1]{{\color{red}\bfseries [#1]}}

\newcommand{\ana}[1]{\todo{ALM: #1}}

\newcommand{\cd}{\ensuremath{\mathit{cd}}\xspace}
\newcommand{\fd}{\ensuremath{\mathit{fd}}\xspace}
\newcommand{\md}{\ensuremath{\mathit{md}}\xspace}

\newcommand{\class}{{\sf class}\xspace}
\newcommand{\extends}{{\sf extends}\xspace}
\newcommand{\return}{{\sf return}\xspace}
\newcommand{\s}{{\sf s}\xspace}

\newcommand{\f}{{\sf f}\xspace}
\newcommand{\g}{{\sf g}\xspace}
\newcommand{\m}{{\sf m}\xspace}

\newcommand{\p}{{\sf p}\xspace}
\newcommand{\vv}{{\sf {v}}\xspace}
\newcommand{\w}{{\sf {w}}\xspace}
\newcommand{\x}{{\sf {x}}\xspace}
\newcommand{\y}{{\sf y}\xspace}
\newcommand{\z}{{\sf z}\xspace}
\newcommand{\D}{{\sf D}\xspace}
\newcommand{\new}{{\sf new}\xspace}
\newcommand{\this}{{\sf this}\xspace}

\newcommand{\ret}{{\sf ret}\xspace}

\newcommand{\low}{\ensuremath{\sf neg}\xspace}
\newcommand{\high}{\ensuremath{\sf pos}\xspace}

\newcommand{\clow}{\textsf{\textbf{neg}}}
\newcommand{\chigh}{\textsf{\textbf{pos}}}

\newcommand{\csens}{\textsf{\textbf{sensitive}}}

\newcommand{\poly}{\ensuremath{\sf poly}\xspace}

\newcommand{\cpoly}{\textsf{\textbf{poly}}}

\newcommand{\readonly}{\ensuremath{\sf readonly}\xspace}

\newcommand{\mutable}{\ensuremath{\sf mutable}\xspace}

\newcommand{\code}[1]{{\sf #1}\xspace}

\newcommand{\rhdri}{\ensuremath{\rhd_{\sc RI}}}

\newcommand{\ol}[1]{\overline{#1}}

\newcommand{\rulename}[1]{\textsc{\scriptsize {(#1)}}\xspace}

\newcommand{\arrow}[1]{\ensuremath{\stackrel{\texttt{#1}}{\rightarrow}}}
\newcommand{\inversearrow}[1]{\ensuremath{\stackrel{\texttt{#1}}{\dashrightarrow}}}
\newcommand{\callarrow}[1]{\ensuremath{\stackrel{(_{#1}}{\rightarrow}}}
\newcommand{\inversecallarrow}[1]{\ensuremath{\stackrel{(_{#1}}{\dashrightarrow}}}
\newcommand{\retarrow}[1]{\ensuremath{\stackrel{)_{#1}}{\rightarrow}}}
\newcommand{\inverseretarrow}[1]{\ensuremath{\stackrel{)_{#1}}{\dashrightarrow}}}
\newcommand{\patharrow}[1]{\ensuremath{\stackrel{#1}{\rightsquigarrow}}}
\newcommand{\approxarrow}{\ensuremath{\stackrel{\texttt{a}}\dashrightarrow}}

\usepackage{microtype}

\setcopyright{acmcopyright}
\copyrightyear{2018}
\acmYear{2018}
\acmDOI{10.1145/1122445.1122456}

\acmConference[Woodstock '18]{Woodstock '18: ACM Symposium on Neural
  Gaze Detection}{June 03--05, 2018}{Woodstock, NY}
\acmBooktitle{Woodstock '18: ACM Symposium on Neural Gaze Detection,
  June 03--05, 2018, Woodstock, NY}
\acmPrice{15.00}
\acmISBN{978-1-4503-XXXX-X/18/06}

\begin{document}

\title[FlowCFL]{FlowCFL: A Framework for Type-based Reachability Analysis in the Presence of Mutable Data}

\author{Ana Milanova}
\affiliation{%
  \institution{Rensselaer Polytechnic Institute, milanova@cs.rpi.edu}
}

  







\renewcommand{\shortauthors}{Milanova, Ana}

\begin{abstract}
Reachability analysis is a fundamental program analysis with a wide variety of applications. We present FlowCFL, a framework for 
type-based reachability analysis in the presence of mutable data. Interestingly, the underlying semantics of FlowCFL is CFL-reachability. 

We make three contributions. First, we define a dynamic semantics that captures the notion of flow commonly used in reachability analysis. 
Second, we establish correctness of CFL-reachability over graphs with \emph{inverse} edges (inverse edges are necessary for the handling of mutable heap data). 
Our approach combines CFL-reachability with \emph{reference immutability} to avoid the addition of certain infeasible
inverse edges and we demonstrate empirically that avoiding those edges results in precision improvement. 
Our formal account of correctness extends to this case as well. Third, we present a type-based reachability analysis and establish 
equivalence between a certain CFL-reachability analysis and the type-based analysis, thus proving correctness of the type-based analysis.   
\end{abstract}

\ccsdesc[500]{Computer systems organization~Embedded systems}
\ccsdesc[300]{Computer systems organization~Redundancy}
\ccsdesc{Computer systems organization~Robotics}
\ccsdesc[100]{Networks~Network reliability}

\keywords{CFL-reachability, reference immutability, type-based analysis}

\maketitle

\input{paper}

\bibliography{library}

\appendix
\input{appendix1}

\end{document}

%% file: paper.tex

\section{Introduction}

Reachability analysis detects flow from \emph{sources} to 
\emph{sinks}. It is a fundamental program analysis technique with a wide variety 
of applications.
One prominent application is taint analysis for Android, which detects flow from 
sensitive sources, such as phone and location data, to untrusted
sinks, such as the Internet~\citep{Arzt:2014PLDI,Ernst:2014CCS,Huang:2015ISSTA}.  

In this paper, we study FlowCFL, a framework for type-based reachability analysis. FlowCFL supports two basic type
qualifiers, \high (positive) and \low (negative). It permits flow from \low variables to \high ones, but forbids
flow from \high variables to \low ones. 
The principal problem is to decide whether there is flow from a \high variable to a \low one. 
Our primary contribution is not the FlowCFL system itself; variants of FlowCFL, both graph-reachability-based and type-based
have been used in program analysis for a long time. There are publicly available implementations, including our own. 
Our contribution is the formal treatment of FlowCFL. We formalize the \emph{notion of flow} in terms of a
dynamic semantics and we use the semantics to construct a correctness argument for the static analyses. 
Another contribution is the interpretation of FlowCFL, a \emph{type-based} reachability analysis, in terms of 
the classical theory of Context-Free-Langauge(CFL)-reachability~\citep{Reps:1995POPL, Reps:1998CFL, Reps:2000TOPLAS}. 

Standard CFL-reachability analysis has two phases.
First, it constructs a graph that represents flow of values from one variable to another; edges are annotated with \emph{call} and \emph{return}
annotations to model call-transmitted dependences and with field \emph{write} and field \emph{read} annotations 
to model heap-transmitted dependences. 
Next, the analysis searches for paths with properly matched call/return and write/read annotations. 
CFL-reachability analysis is a highly precise flow analysis. It has a long history~\cite{Reps:1995POPL, Reps:2000TOPLAS} 
and it is still actively studied and actively used in program analysis;~\cite{Xu:2009ECOOP, Zhang:2017POPL, Chatterjee:2018POPL, Spath:2019POPL, Lu:2019OOPSLA} are recent works among many other works. 
An important concept in CFL-reachability analysis is the concept of the \emph{inverse edge}~\cite{Sridharan:2005OOPSLA,Sridharan:2006PLDI}, which is necessary for the handling of mutable heap data. At assignments, e.g., at {\sf x = y}, the analysis adds the expected \emph{forward
edge} from \y to \x that represents flow from \y to \x, however, it also adds an \emph{inverse edge} from \x to \y thus constructing a \emph{bidirectional CFL-reachability graph} $G_{\sc BI}$. 
We discuss inverse edges, the bidirectional CFL-reachability graph, and mutable data in detail in~\secref{sec:CFL}.
The concept of the inverse edge, which becomes more involved once 
we consider call/return and write/read annotated edges and paths, has not been formalized. 
The question \emph{``Does CFL-reachability over $G_{\sc BI}$ capture all program flows?''} 
has not been answered formally. We consider an answer to this question.

Additionally, we consider a graph that avoids adding certain inverse edges based on knowledge of reference immutability. We denote this graph by $G_{\sc RI}$.
Given an \emph{immutable reference} \x, there is no need to add inverse 
paths that originate at \x. There is substantial precision improvement for reachability over $G_{\sc RI}$ compared to reachability over $G_{\sc BI}$ 
as demonstrated in earlier work~\cite{Milanova:2013FTfJP,Zhang:2017POPL} and confirmed by experiments we run for this paper. 
Our formal treatment extends to this case. We answer the following question as well: \emph{``Does CFL-reachability over $G_{\sc RI}$ capture all program flows?''}.

Returning to FlowCFL, the type-based analysis uses the \high, \low, and \poly type qualifiers and a set of typing rules to model reachability. Although not as widespread 
as CFL-reachability, type-based reachability analysis has been used in a number of existing works, e.g.,~\cite{Shankar:2001USENIXSecurity,Sampson:2011PLDI,Huang:2015ISSTA}. 
Type-based analysis conveniently models reachability problems in different domains, including taint analysis, approximate computing, and secure computation. We establish equivalence between a certain type-based reachability analysis and a certain CFL-reachability analysis over $G_{\sc RI}$, thus establishing correctness
of the type-based analysis.


This paper makes the following contributions:

\begin{itemize}

\item We present a dynamic semantics that formalizes the notion of \emph{flow} commonly used in CFL-reachability and type-based reachability.

\item We prove that CFL-reachability over $G_{\sc BI}$ captures all run-time flows. Our treatment extends to reachability over 
$G_{\sc RI}$ which avoids adding certain edges based on knowledge of reference immutability. We present experiments that show
substantial precision improvement in taint analysis for Android over $G_{\sc RI}$ compared to analysis over $G_{\sc BI}$. Our experiments are 
in line with earlier work that has shown the importance of reducing the number of inverse edges~\cite{Milanova:2013FTfJP,Zhang:2017POPL}.

\item We establish equivalence between a type-based reachability analysis and a CFL-reachability analysis, thus proving
correctness of the type-based analysis.

\end{itemize}

The rest of the paper is organized as follows. \secref{sec:overviewAndApplications}
presents an overview of FlowCFL 
and briefly discusses applications of FlowCFL. 
\secref{sec:dynamicSemantics} presents the dynamic semantics of flows.
\secref{sec:CFL} presents CFL-reachability, $G_{\sc BI}$, reference immutability, and 
the construction of $G_{\sc RI}$.
\secref{sec:soundness} details the correctness argument.
\secref{sec:typeBasedSemantics} presents the type-based analysis,
and~\secref{sec:equivalence} establishes equivalence between the type-based
and CFL-based analyses. \secref{sec:relatedWork} discusses related
work and \secref{sec:conclusion} concludes.

\section{Overview and Applications}
\label{sec:overviewAndApplications}

\subsection{Overview of FlowCFL}
\label{sec:overview}

In a typical setting, reachability analysis reasons about flow from \emph{sources} to \emph{sinks}.
FlowCFL assigns type \emph{qualifier} to variables and fields. 
There are two basic qualifiers: \high, which denotes sources, and \low, which denotes sinks. We have
\[
\low <: \high 
\]
where $q_1 <: q_2$ denotes $q_1$ is a subtype of $q_2$. ($q$ is also a subtype of itself $q <: q$.) 
Therefore, it is allowed to assign a \low variable to a \high one, i.e., a \low variable can flow to a \high one:

\begin{center}
\begin{minipage}{5cm}
\begin{lstlisting}[numbers=none]
  %\clow% String n = ...; %\vspace{0.1pt}%
  %\chigh% String p = n; 
\end{lstlisting}
\end{minipage}
\end{center}

However, it is not allowed to assign a \high variable to a \low one, i.e., a \high variable cannot flow to a \low one:

\begin{center}
\begin{minipage}{5cm}
\begin{lstlisting}[numbers=none]
  %\chigh% String p = ...; %\vspace{0.1pt}%
  %\clow% String n = p;    // error!
\end{lstlisting}
\end{minipage}
\end{center}

Note that this is the natural subtyping. Such subtyping is unsafe
in the presence of mutable references~\cite{Myers:1997POPL,Sampson:2011PLDI}
and systems use \emph{equality}, which is akin to the inverse edges in CFL-reachability.
FlowCFL leverages reference immutability (e.g., ReIm~\cite{Huang:2012OOPSLA}, Javari~\cite{Tschantz:2005OOPSLA}) to allow for 
safe but limited subtyping. 

Once the sources and/or sinks are given, FlowCFL \emph{infers} qualifiers for
the rest of the variables. Roughly, if a source flows to a variable \x, then \x is 
\high; if a variable \y flows to a sink, then \y is \low. If inference fails, i.e., reports \emph{error(s)}, 
then there may be a leak from a source to a sink. Otherwise, it is guaranteed that there is no flow 
from a source to a sink.

FlowCFL is context-sensitive (i.e., polymorphic) as illustrated by the following example. 
We elaborate on context sensitivity in~\secsref{sec:CFL}-\ref{sec:typeBasedSemantics}.

\begin{center}
\begin{minipage}{5cm}
\begin{lstlisting}
%\cpoly% String id(%\cpoly% String p) {
    return p;
}
%\chigh% String source = ...;
%\chigh% String x = id(source);
   
%\clow% String y = ...;
%\clow% String sink = id(y);   
\end{lstlisting}
\end{minipage}
\end{center}

In the above example, the identity function {\sf id} is context-sensitive. {\sf id} is
interpreted as \high{} in line 5 and it is interpreted as \low{} in line 8. FlowCFL precisely propagates
{\sf source} to \x but not to {\sf sink}; it propagates {\sf sink} back to \y but not to {\sf source}.
A context-insensitive system rejects the program as it merges flow through {\sf id}
and imprecisely decides that there is flow from {\sf source} to {\sf sink}.



From a practical point of view, FlowCFL supports two different settings of
the problem. The \emph{negative setting} assumes a set of initial sink annotations and propagates those sinks backwards, i.e., against the direction of the flow, 
towards program variables. 
Unaffected variables remain positive. The more precise the analysis, the fewer variables become \low
and a larger number of variables remain \high.
The \emph{positive setting} assumes initial source annotations and propagates those sources forward. 
The well-known ``taint-analysis'' problem, which entails annotations on both 
sources and sinks, can be cast in either of the settings. FlowCFL can run in either the negative or positive setting; if it detects 
a conflict, i.e., a variable is annotated \high{} but is inferred \low{} (or in the positive setting, 
it is annotated \low{} but is inferred \high{}), it reports an error. 


\begin{figure}[t]
  \begin{tabular}{lll}
  \begin{minipage}{5.25cm}
    \begin{lstlisting}
public class Data { 
  String secret;
  void set(String p) { 
     this.secret = p; 
  }
  String get() { 
     return this.secret; 
  } 
}
  \end{lstlisting}
  \end{minipage}
  
  &
  
  &
  
  \begin{minipage}{8cm}
    \begin{lstlisting}
public class FieldSensitivity2 extends Activity {
  protected void onCreate(Bundle b) {
    Data dt = new Data(); 
    TelephonyManager tm = (TelephonyManager) getSystemService("phone");
    %\chigh% String sim = tm.getSimSerialNumber(); 
    dt.set(sim);    
    SmsManager sms = SmsManager.getDefault();
    %\clow% String sg = dt.get();
    sms.sendTextMessage("+123",null,sg,null,null);
  }
}
    \end{lstlisting}
  \end{minipage}
  \end{tabular}
  \caption{FieldSensitivity2 is rephrased from DroidBench~\cite{Fritz:2013TR,Arzt:2014PLDI}.
  {\sf getSimSerialNumber} in line 5 in {\sf FieldSensitivity2} retrieves sensitive telephony information 
  and its return value is a source. The parameter of {\sf sendTextMessage} in line 9 is a sink.  
  There is flow from source {\sf sim} to sink {\sf sg} through the {\sf Data} container and FlowCFL reports an error. 
  We note that in actual implementations of taint analysis, there are no annotations in app code only in the Android SDK;
  we have annotated {\sf sim} as \high and {\sf sg} as \low in the above code purely for illustration purposes.
} 
\label{fig:fieldsensitivity}
\end{figure}

\subsection{Applications}


There is a wide variety of applications of FlowCFL. One prominent example, taint analysis for Android reasons about flow of sensitive data 
(e.g., phone data, location data) to untrusted sinks (e.g., the Internet, Sms texts). 
\cite{Huang:2015ISSTA} and~\cite{Ernst:2014CCS}, among others, describe type-based taint analysis that are instances of FlowCFL. 
\figref{fig:fieldsensitivity} illustrates taint analysis for Android with FlowCFL. 
FlowCFL decides that there is flow from positive {\sf sim}, the device SIM serial number (SSN), to negative {\sf sg}, the body of the text message. 
Assuming inference in the negative setting, the analysis determines that {\sf sim} is \low{}, 
which clashes with the designation of {\sf sim} as source and {\sf sim}'s \high type. The analysis issues an error that captures that source {\sf sim} flows to sink {\sf sg}. 

There are many instances of FlowCFL and different instances generally demand different settings. Type systems that underpin approximate computing~\cite{Sampson:2011PLDI, Carbin:2013OOPSLA, Bornholt:2014ASPLOS, Holt:2016SOCC} are instances of FlowCFL. E.g., EnerJ~\cite{Sampson:2011PLDI} can be expressed as an instance of FlowCFL in the negative setting. Secure computation, where reachability analysis can partition a program into a secure (and expensive) partition and a plaintext (and inexpensive) partition, is another area of application; the analysis can be expressed as an instance of FlowCFL in the positive setting~\cite{Dong:2016PPPJ}.
We have implemented EnerJ~\cite{Sampson:2011PLDI}, Rely~\cite{Carbin:2013OOPSLA}, DroidInfer~\cite{Huang:2015ISSTA}, and JCrypt~\cite{Dong:2016PPPJ} as instances of FlowCFL (Appendix~\ref{app:applications} describes the instantiations). 
The wide variety of applications motivates our study of FlowCFL and its connection to CFL-reachability.

\subsection{Overview of CFL-reachability}
\label{sec:overviewCFL}

FlowCFL is a type-based analysis but its underlying semantics is the classical CFL-reachability analysis~\cite{Reps:1995POPL,Reps:1998CFL,Reps:2000TOPLAS}.
CFL-reachability analysis proceeds in two phases. First, it constructs a flow graph representation of the program. Second, it reasons about reachability over the graph. 
Throughout the paper, we will work with the example in~\figref{fig:runningExample}, which is a rephrase of the {\sf FieldSensitivity} example in~\figref{fig:fieldsensitivity}.
CFL-reachability constructs the graph shown below. Annotations $(_i$ and $)_i$ model call-transmitted dependences. For example, edge ${\sf a} \callarrow{6} \p$ represents that at call site 6 in {\sf main} {\sf a} flows to parameter \p of {\sf set}, and edge ${\sf ret} \retarrow{7} {\sf b}$ represents that at call site 7 \ret of {\sf get} flows to {\sf b}. Annotations $w_\f$ and $r_\f$ model heap-transmitted dependences. $\p \arrow{$w_\f$} \this_{\sf set}$ models flow of \p into field \f of $\this_{\sf set}$, and $\this_{\sf get} \arrow{$r_\f$} \ret$ models read of field \f from $\this_{\sf get}$ into \ret. 
\begin{figure}[t]
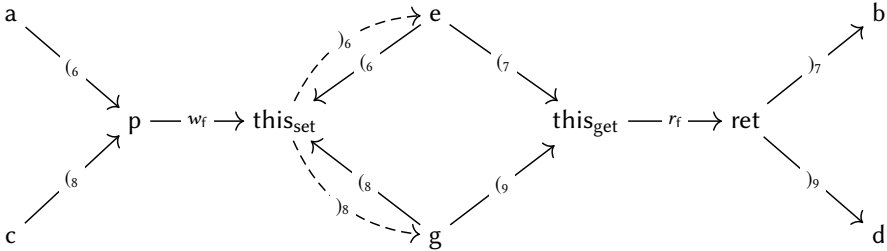

  \begin{tabular}{lll}
  \begin{minipage}{5.25cm}
    \begin{lstlisting}
public class A { 
  B f;
  void set(A this, B p) { 
     this.f = p; 
  }
  B get(A this) { 
     ret = this.f;
     return ret;  
  } 
}
  \end{lstlisting}
  \end{minipage}

  & 
  
  &
  
  \begin{minipage}{8cm}
    \begin{lstlisting}
public class C {
  public static void main(...) {
    A e = new A(); 
    A g = new A();
    ...
    e.set(a);
    b = e.get();
    g.set(c);
    d = g.get();
  }
}
    \end{lstlisting}
  \end{minipage}
  \end{tabular}
  \caption{Running example. {\sf a} in line 6 in {\sf main} flows to {\sf b} in line 7; {\sf c} in line 8 flows to {\sf d} in line 9. We make parameter \this explicit.
} 
\label{fig:runningExample}
\end{figure}

\[
\begin{tikzcd}[cramped, sep=large]
{\sf a} \arrow[dr, "{(_6}" description] & & & {\sf e} \arrow[dr, "{(_7}" description] \arrow[dl, "{(_6}" description] & & & {\sf b} \\
& \p \arrow[r, "{w_\f}" description] & \this_{\sf set} \arrow[ur, bend left, dashed, "{)_6}" description] \arrow[dr, bend right, dashed, "{)_8}" description] & & \this_{\sf get} \arrow[r, "{r_\f}" description] & \ret \arrow[ur, "{)_7}" description] \arrow[dr, "{)_9}" description]&  \\
{\sf c} \arrow[ur, "{(_8}" description] & & & {\sf g} \arrow[ur, "{(_9}" description] \arrow[ul, "{(_8}" description]  & & & {\sf d} \\
\end{tikzcd} 
\]
The principal problem is to decide whether there is flow from one node to another and CFL-reachability analysis makes use of the call/return and write/read annotations to make the decisions. In our example, there is flow from {\sf e} to {\sf b} because the call and return annotations, $(_7$ and $)_7$ respectively, match. However, there is no flow from {\sf e} to {\sf d} because call annotation $(_7$ and return annotation $)_9$ do not match; they denote two distinct calls. Analogously, field write and field read annotations have to match, there is flow from {\p} to {\ret} because $w_\f$ and $r_\f$ denote a write and a read of the same field \f. The analysis decides that there is flow from {\sf a} to {\sf b} and from {\sf c} to {\sf d}, however there is no flow from $\sf a$ to $\sf d$ or from $\sf c$ to $\sf b$.

There are two notable points. The first point concerns inverse edges. In the example, there is a forward edge from {\sf e} to ${\sf this}_{\sf set}$ (solid: $\rightarrow$) and there is an inverse edge from  ${\sf this}_{\sf set}$ to {\sf e} (dashed: $\inversearrow{}$) that reverses the direction of flow and the annotation; call annotation $(_6$ becomes return annotation $)_6$. The forward edge is a natural addition to the graph representing flow of receiver {\sf e} to ${\sf this}_{\sf set}$. The inverse edge, however, is unnatural but it is necessary to discover the path from {\sf a} to {\sf e} and then to {\sf b} as the \emph{mutation}, i.e., update of $\this_{\sf set}$ reverses flow. Standard reachability analysis adds an inverse edge for every forward edge (e.g., \cite{Zhang:2017POPL}, \cite{Lu:2019OOPSLA}); in our example, every solid edge would have had a corresponding dashed edge in the standard $G_{\sc BI}$. Our analysis takes into account reference immutability information and adds only the minimal number of inverse edges necessary to decide flow correctly; the graph shown is the $G_{\sc RI}$. One key problem we address is to show that $G_{\sc BI}$, and more interestingly $G_{\sc RI}$, indeed capture all run-time flows. We define a dynamic semantics that formalizes run-time flows, in our example, the meaning of ``\p in context of invocation of {\sf set} in line 6 flows to \ret in context of invocation of {\sf get} in line 7'' (\secref{sec:dynamicSemantics}). We proceed to define the construction of $G_{\sc BI}$ and $G_{\sc RI}$ (\secref{sec:CFL}) and argue soundness, i.e., that $G_{\sc RI}$ does represent all run-time flows (\secref{sec:soundness}). The second point concerns paths with interleaved call/return and write/read annotations. For example, the path from {\sf a} to {\sf b} involves matching call/return annotations, $(_6$ and $)_6$, as well as $(_7$ and $)_7$, and separately, matching write/read annotations $w_\f$ and $r_\f$. Exact reasoning over such paths is undecidable~\cite{Reps:2000TOPLAS}. We present a certain approximate reachability analysis over $G_{\sc RI}$ and show that type-based FlowCFL is equivalent to that analysis (\secref{sec:equivalence}). We interpret the seemingly different type-based FlowCFL in terms of CFL-reachability (\secref{sec:typeBasedSemantics}).

\section{Dynamic Semantics}
\label{sec:dynamicSemantics}

In this section we formalize the notion of flow in terms of a dynamic semantics. 
We restrict our core language to a ``named form'' in the style of 
Vaziri et al.~\cite{Vaziri:2010ECOOP,Dolby:2012TOPLAS}. The language 
models Java with the syntax in~\figref{fig:syntax}, where the results of instantiations, 
field accesses, and method calls are immediately stored in a variable. Without loss of 
generality, we assume that methods have parameter \this, and exactly one other formal 
parameter. 

\begin{figure}[t]
\small
\[
\begin{array}{l@{~~~~~}l}
  \begin{array}{l@{~}l@{~~~}l}
 \cd & ::= \class\; {\sf C} \; \extends\; \D \;\{\ol\fd\; \ol\md\}& \mathit{class} \\
 \fd &::= t \;\f & \mathit{field} \\
 \md &::= t\;{\m}(t \;{\this}, \; t\;\x) \; \{\; \ol{t\;\y}\;\s;\;\return\;\y \; \} 
                  & \mathit{method} \\ 
  \s &::= \s;\s \mid  \x=\new\; t \mid \x=\y  \mid \x=\y.\f  
  \mid \y.\f=\x \mid \x=\y.\m(\z) & \mathit{statement} \\ 
   t &::= q \; {\sf C} & \mathit{qualified \; type}\\
   q &::= \high \mid \poly \mid \low & \mathit{FlowCFL \; qualifier}\\

  \end{array}
\end{array}
\]
\caption{Syntax. {\sf C} and {\D} are class names, \code{f} is a field name,
\code{m} is a method name, and \code{x}, \code{y}, and \code{z} are names of local variables,
formal parameters, or parameter \code{this}. As in the code examples,
\code{this} is explicit. 
The syntax separates object creation from initialization (i.e., {\sf x = new t()} 
becomes {\sf x = new t; x.init()}).
} 
\label{fig:syntax}
\end{figure}

\subsection{Stack Contexts and Chains}
\label{sec:framesAndChains}


\emph{Stack contexts} describe stack configurations at a point of program execution; as expected, they help formalize 
run-time local variables, such as the following: ``\p in context of invocation of {\sf set} in line 6``. $\langle {\sf main}, f_1, f_2 ... f_n \rangle$ 
is the stack made up of ${\sf main}$, followed by frame identifier $f_1$ corresponding to some callee $m_1$ in {\sf main}, followed by $f_2$ 
for some callee $m_2$ in $m_1$, etc, with frame $f_n$ at the top of the stack. Each $f_i$ has a unique identifier---if, say, $m_2$ is 
called again from the same call site in $m_1$, that would entail a new frame and frame identifier at runtime. We use $A, B, C,...$ to denote 
stack contests. Local variables, naturally, are characterized by their stack context: we write $\x^A$ to denote local variable \x in context $A$. 

In our running example in~\figref{fig:runningExample}, we have stack contexts $\langle {\sf main}, f_1 \rangle$ and $\langle {\sf main}, f_2 \rangle$ 
where $f_1$ corresponds to the frame of {\sf set} invoked in line 6, and $f_2$ corresponds to the frame of {\sf set} invoked in line 8. 
Variables in {\sf set} are characterized by their stack context: we write $\p^{\langle {\sf main}, f_1 \rangle}$, $\this^{\langle {\sf main}, f_1 \rangle}$, etc.

The notion of the \emph{chain} is essential in our treatment. Informally, there is a chain from $\x^A$ to $\y^B$, denoted by $(\x^A, \y^B)$, if $\x^A$ flows to $\y^B$.
In~\figref{fig:runningExample} there is a chain from $\p^{\langle {\sf main}, f_1 \rangle}$ to $\ret^{\langle {\sf main}, f_3 \rangle}$
where $f_3$ is the frame that corresponds to the invocation of {\sf get} in line 7. Similarly, there are chains $({\sf a}^{\langle {\sf main} \rangle}, \ret^{\langle {\sf main}, f_3 \rangle})$, 
 $({\sf a}^{\langle {\sf main} \rangle}, {\sf b}^{\langle {\sf main}\rangle})$ among others. Chains are represented in $G_{\sc RI}$ by appropriately annotated paths. E.g., 
 chain $({\sf a}^{\langle {\sf main} \rangle}, \ret^{\langle {\sf main}, f_3 \rangle})$ is represented by path 
 \[{\sf a} \callarrow{6} \p \arrow{$w_\f$} \this_{\sf set} \retarrow{6} {\sf e} \callarrow{7} \this_{\sf get} \arrow{$r_\f$} \ret \] 
 The unmatched call annotation $(_7$ represents that {\sf a} flows into 
 frame $f_3$ where $f_3$ maps to call site 7 in the abstract. The string with unmatched call $(_7$ captures in the abstract the ``difference'' between contexts 
 ${\langle {\sf main} \rangle}$ and $\langle {\sf main}, f_3 \rangle$.

The semantics is a standard small-step dynamic semantics extended with the treatment of chains. We write $\llbracket s \rrbracket(A, \mathbb{C}, \mathbb{S}, \mathbb{H}) = \mathbb{C'}, \mathbb{S'}, \mathbb{H'}$ to model execution of statement $s$ in context $A$ and its effect on chains $\mathbb{C}$, stack $\mathbb{S}$, and heap $\mathbb{H}$. Here $\mathbb{C}$ 
is a map from variables $\y^B$ to \emph{sets of sources} of chains that end at $\y^B$. We say $(\x^A,\y^B) \in \mathbb{C}$ iff $\y^B \in \mathit{Dom}(\mathbb{C})$
and $\x^A \in \mathbb{C}(\y^B)$. $\mathbb{S}$ and $\mathbb{H}$ are the standard
maps from variables to objects $o$ (map $\mathbb{S}$), and from object/field tuples $o.\f$ to objects $o'$ (map $\mathbb{H}$). Below we discuss 
the semantics of individual statements. 

An assignment {\sf x = y} in context $A$ records chains that end at $\x^A$. There is a new chain $(v,\x^A) \in \mathbb{C'}$ for every chain $(v,\y^A) \in \mathbb{C}$, 
and there is a chain $(\y^A, \x^A) \in \mathbb{C'}$ that accounts for the flow from \y to \x. The transition on the stack is standard: $\x^A$ points to the object that $\y^A$ points to
and the heap remains the same:
\[
\mbox{ASSIGN} \quad \llbracket {\sf x = y} \rrbracket (A, \mathbb{C}, \mathbb{S}, \mathbb{H}) = \mathbb{C}[\x^A \leftarrow \mathbb{C}(\y^A)\cup \{ \y^A\}], \mathbb{S}[\x^A \leftarrow \mathbb{S}(\y^A)], \mathbb{H}
\]

Field write {\sf x.f = y} records chains $(v, o.\f)$ and field read ${\sf y' = x'.f}$ references those chains to record $(v, (\y')^A)$:
\[
\begin{array}{llll}
\mbox{WRITE} & \llbracket {\sf x.f = y} \rrbracket (A, \mathbb{C}, \mathbb{S}, \mathbb{H}) & = & \mathbb{C}[o.\f \leftarrow \mathbb{C}(\y^A)\cup \{ \y^A\}], \mathbb{S}, \mathbb{H}[o.\f \leftarrow o'] \\
& & & \quad\quad \mbox{ where } o = \mathbb{S}(\x^A) \mbox{ and } o' = \mathbb{S}(\y^A) \\ 
\mbox{READ} & \llbracket {\sf y' = x'.f} \rrbracket (A, \mathbb{C}, \mathbb{S}, \mathbb{H}) & = & \mathbb{C}[(\y')^A \leftarrow \mathbb{C}(o.\f)], \mathbb{S}[(\y')^A \leftarrow o'], \mathbb{H} \\
& & & \quad\quad \mbox{ where } o = \mathbb{S}((\x')^A) \mbox{ and } o' = \mathbb{H}(o.\f) \\
\end{array}
\]
Note that we do not record $o.\f$ as a chain source in READ. It is an invariant of $\mathbb{C}$ that the values in map $\mathbb{C}$ are sets that contain only variables, e.g., $\x^A$. 
We do record $o.\f$ as chain target in WRITE because it serves as an intermediary in the chain from \y in {\sf x.f = y} to $\y'$ in ${\sf y' = x'.f}$. 
The semantics of chains elides heap objects, just as the static flow graph $G_{\sc RI}$ does (recall the graph in~\secref{sec:overviewCFL}). 
The goal is to establish a connection between the concrete domain of chains and the abstract domain of annotated paths in the flow graph.

Allocation ${\sf x = new } \; o$ creates the trivial chain $(\x^A,\x^A)$:
\[
\begin{array}{llll}
\mbox{ALLOC} & \llbracket {\sf x = new } \; o \rrbracket (A, \mathbb{C}, \mathbb{S}, \mathbb{H}) & = & \mathbb{C}[\x^A \leftarrow \{\x^A\}], \mathbb{S}[\x^A \leftarrow o], \mathbb{H}[o.\f \leftarrow {\sf null}] \\
& & & \quad\quad \mbox{ where } o \mbox{ is a fresh object}
\end{array}
\]

Calls and returns are standard. A call entails a fresh frame identifier $f$, which is appended to context $A$ to form the new context $A\!\oplus\! f$; calls record chains that  
reflect the standard flow from actuals to formals:
\[
\begin{array}{llll}
\mbox{CALL} & \llbracket {\sf x = y.m(z)} \rrbracket(A, \mathbb{C}, \mathbb{S}, \mathbb{H}) & = &  \mathbb{C}[\this^{A \oplus f} \leftarrow \mathbb{C}(\y^A) \cup \{\y^A\}][\p^{A \oplus f} \leftarrow \mathbb{C}(\z^A) \cup \{\z^A\}], \\
& & & \mathbb{S}[\this^{A \oplus f} \leftarrow \mathbb{S}(\y^A)][\p^{A \oplus f} \leftarrow \mathbb{S}(\z^A)], \mathbb{H} \\
& & & \qquad \qquad \mbox{ where } f \mbox{ is a fresh frame identifier} \\
\end{array}
\]
A return from context $A\!\oplus\! f$ creates chains with target $\x^A$; these chains are due to the standard flow from $\ret^{A\oplus f}$ to the left-hand side of the return assignment $\x^A$:
\[
\begin{array}{llll}
\mbox{RET} & \llbracket {\sf x = y.m(z)} \rrbracket(A \!\oplus \! f, \mathbb{C}, \mathbb{S}, \mathbb{H}) & = &  \mathbb{C}[\x^A \leftarrow \mathbb{C}(\ret^{A\oplus f}) \cup \{\ret^{A\oplus f}\}], \mathbb{S}[\x^A \leftarrow \mathbb{S}(\ret^{A \oplus f})], \mathbb{H} \\
\\
\end{array}
\]


The following lemma allows us to express the points-to relation entailed by $\mathbb{S}$ and $\mathbb{H}$ in terms of the reachability relation entailed by $\mathbb{C}$. 
If a variable $\x^A$ points to some object $o$, then there is a chain $(\w^B,\x^A)$ in $\mathbb{C}$ where $\w^B$ is the local variable at the left-hand side of the allocation 
site of $o$. If we have ${\sf x.f = y}$ in context $A$ followed by ${\sf y' = x'.f}$ in context $A'$, where \x and $\x'$ point to the same object $o$, 
then there is flow from \y to $\y'$. The flow is expressed via paths that capture chains $(\w^B,\x^A)$ and $(\w^B, (\x')^{A'})$. The key idea is that in the abstract, 
we have a path from $\w$ to $\x$ and an \emph{inverse path} from $\x$ to $\w$. The inverse path $\x \patharrow{} \w$ ``combines'' with the path 
$\w \patharrow{} \x'$ which leads to a path $\x \patharrow{} \x'$ and subsequently $\y \patharrow{} \y'$. Recall~\figref{fig:runningExample} and the accompanying graph:
\[{\sf this}_{\sf set} \inverseretarrow{6} {\sf e} \mbox{ and } {\sf e} \callarrow{7} {\sf this}_{\sf get} \mbox{ give rise to } \p \arrow{$w_\f$} {\sf this}_{\sf set} \inverseretarrow{6} {\sf e} \callarrow{7} {\sf this}_{\sf get} \arrow{$r_\f$} {\sf ret}\] which abstracts the chain from $\p^{\langle {\sf main}, f_1\rangle}$ to $\ret^{\langle {\sf main}, f_3\rangle}$. We elaborate later in the paper.

{\lemma 
For every state $\mathbb{C}, \mathbb{S}, \mathbb{H}$ and every object $o \in \mathbb{H}$
\begin{itemize}
\item $\mathbb{S}(\x^A) = o \quad \Rightarrow \quad (\w^B,\x^A) \in \mathbb{C}$ and 
\item $\mathbb{H}(o'.\f) = o \quad \Rightarrow \quad (\w^B,o'.\f) \in \mathbb{C}$
\end{itemize}
where $\w = \mathit{new} \; {\sf C}$ in context $B$ is the creation site of $o$. 
\label{lemma:chainLemma}
}

\begin{proof} Given transition $\llbracket s \rrbracket(A,\mathbb{C},\mathbb{S},\mathbb{H}) = \mathbb{C'},\mathbb{S'},\mathbb{H'}$ 
we show, via case-by-case analysis, that if the lemma holds on $\mathbb{C},\mathbb{S},\mathbb{H}$ then it holds on $\mathbb{C'},\mathbb{S'},\mathbb{H'}$. 
Consider ALLOC. The only change is in $\mathbb{S}$ where $\mathbb{S'}(\x^A)$ now points to a fresh heap object $o$. $\mathbb{C'}$
adds $(\w^A, \w^A)$ which establishes the lemma. The remaining points-to relations have not changed
and the inductive hypothesis entails the lemma. The rest of the statements follow analogously.

\end{proof}

\subsection{Operations on Contexts}
\label{sec:frame}

Next we define several useful operations on stack contexts. $A-B$ is defined when $B$ is a prefix of $A$; it removes $B$ from $A$. 
$A \le B$ is true if and only if $A$ is a prefix of $B$. 
$\Delta A B$ denotes the ``difference'' between context $A$ and context $B$. At the level of the dynamic semantics, $\Delta AB$ is
defined as the tuple $(A-D, B-D)$, where $D$ is the longest common prefix of $A$ and $B$. Such a prefix clearly exists, 
in the worst case it is {\sf main}.

Returning to~\figref{fig:runningExample}, 
consider the flow from \p in context ${\langle {\sf main}, f_2 \rangle}$ of {\sf set}, to \ret in context ${\langle {\sf main}, f_3 \rangle}$ 
of {\sf get} (recall that $f_3$ is the frame invoked in line 7). We have:
\[\Delta \langle {\sf main}, f_2 \rangle \langle {\sf main}, f_3 \rangle = (\langle f_2\rangle, \langle f_3\rangle) \]

Informally, the first term in the tuple is the sequence of \emph{returns} and the second term is the sequence of 
\emph{calls} that happen when state transitions from $A$ to $B$. In the example, stack state transitions from $\langle {\sf main}, f_2 \rangle$
to $\langle {\sf main}, f_3 \rangle$, by first returning from $f_2$ into {\sf main} then calling into $f_3$ from {\sf main}.

In~\secref{sec:soundness} we define an abstraction function over stack contexts and $\Delta AB$
that helps establish that for every chain from $\x^A$ to $\y^B$, $G_{\sc BI}$ and $G_{\sc RI}$ contain appropriately annotated paths from \x to \y.
As stated earlier, a key difficulty arises in the reasoning about inverse edges.

\section{CFL-reachability}
\label{sec:CFL}

\secref{sec:bidirectionalFlowGraph}
describes the construction of $G_{\sc BI}$. \secref{sec:imprecision} argues that there is inherent imprecision in $G_{\sc BI}$. 
\secref{sec:referenceImmutability} discusses reference immutability and \secref{sec:referenceImmutabilityGraph}
describes the construction of $G_{\sc RI}$ based on knowledge of reference immutability.


\subsection{Bidirectional Flow Graph $G_{\sc BI}$}
\label{sec:bidirectionalFlowGraph}


As it is customary for CFL-reachability, we build a static \emph{flow graph}
that represents data dependences between variables. The nodes in the graph are 
(context-insensitive) program variables, e.g., \x, \y, \this. The edges capture flow from one variable to another
and paths capture dynamic chains as defined in~\secref{sec:dynamicSemantics}. The standard approach in the presence of mutable data is to build a 
\emph{bidirectional} flow graph (as in~\cite{Zhang:2017POPL, Sridharan:2006PLDI, Xu:2009ECOOP, Chatterjee:2018POPL, Spath:2019POPL, Sampson:2011PLDI, Shankar:2001USENIXSecurity}
among other works)
where \emph{inverse edges} handle updates safely. We call this graph $G_{\sc BI}$. Below we describe the semantics 
of $G_{\sc BI}$ construction. Solid arrows $\arrow{}$ denote forward edges, and dashed arrows $\inversearrow{}$ denote inverse edges.

An assignment statement contributes \emph{direct} (i.e., intraprocedural) 
edges as follows:
\[\mbox{ASSIGN}\quad \llbracket \code{x = y}\rrbracket (G_{\sc BI}) = G_{\sc BI} \cup \{ \y \arrow{d} \x \} \cup \{ \x \inversearrow{d} \y \} \]



A field write statement \code{x.f = y} contributes a forward edge from \y to \x annotated with $w_\f$
and an inverse edge from \x to \y annotated with $r_\f$:
\[\mbox{WRITE }\quad \llbracket \code{x.f = y} \rrbracket (G_{\sc BI}) = G_{\sc BI} \cup \{ \y \arrow{$w_\f$} \x \} \cup \{ \x \inversearrow{$r_\f$} \y \} \]
The meaning of the forward edge is that \y flows (is written) into field \f of \x. The corresponding inverse edge reverses 
the direction of the flow and the field annotation, denoting that field \f of \x is read into \y.
Similarly, a field read statement $\code{y' = x'.f}$ contributes
\[\mbox{READ }\quad \llbracket \code{y' =  x'.f}\rrbracket(G_{\sc BI}) = G_{\sc BI} \cup \{ \x' \arrow{$r_\f$} \y' \} \cup \{ \y' \inversearrow{$w_\f$} \x' \} \]


The following example illustrates once again the need for inverse edges. From now on, we will underline sinks in the graphs to improve readability.
\begin{center}
\begin{tabular}{lll}
\begin{minipage}{4cm}
\begin{lstlisting} 
x = y;
A a = ... ; 
x.f = a;  
%\clow% A b = y.f;  
\end{lstlisting} 
\end{minipage}

&

$\Rightarrow$

&

\begin{tikzcd}[cramped, sep=large]
{\sf a} \arrow[r, "{w_\f}" description] & \x \arrow[r, dashed] \arrow[l, bend left, dashed, "{r_\f}" description] & \y \arrow[r, "{r_\f}" description] \arrow[l, bend right] & {\sf \underline{b}} \arrow[l, bend left, dashed, "{w_\f}" description]
\end{tikzcd} \\

\end{tabular}
\end{center}
Inverse edge $\x \inversearrow{d} \y$ is necessary to establish the path from {\sf a} to {\sf b}.

A method call (method entry) creates the expected forward \emph{call} edges from actual arguments
to formal parameters and the inverse \emph{return} edges:
\[\mbox{CALL}\quad \llbracket i\!: \code{x = y.m(z)} \rrbracket(G_{\sc BI}) = G_{\sc BI} \cup \{ \y \callarrow{i} \this \} \cup \{ \z \callarrow{i} \p \} \cup \{ \this \inverseretarrow{i} \y \} \cup \{ \p \inverseretarrow{i} \z \} \]
The standard CFL-reachability annotation $(_i$ marks call entry at call site $i$.
A method return (exit) 
creates a \emph{return} edge from the return value to the left-hand-side 
of the call assignment, plus the inverse \emph{call} edge:
\[\mbox{RET}\quad \llbracket i\!:\code{x = y.m(z)}\rrbracket (G_{\sc BI}) = G_{\sc BI} \cup \{ \ret \retarrow{i} \x \}\cup\{ \x \inversecallarrow{i} \ret \} \]

The CFL-reachability problem is to decide whether there is a path from \x to \y in $G_{\sc BI}$ with properly matched call/ret annotations,  
and properly matched write/read annotations. Note the arbitrary interleaving of $(,)$ and $w,r$ annotations.  
Due to recursion in both call-transmitted and heap-transmitted dependences, the CFL-reachability problem is 
undecidable~\cite{Reps:2000TOPLAS} and analyses have to adopt approximations. 
One approximation essentially replaces all $(_i$ and $)_j$ with {\sf d} annotations and leaves all $w,r$ annotations, 
thus treating call-transmitted dependences context-insensitively and heap-transmitted dependences fully precisely. 
This approach is known as CIFS (contest-insensitive, field-sensitive) CFL-reachability~\cite{Zhang:2017POPL}. 
Another approximation replaces all $w_\f$ and $r_\g$ annotations with {\sf d} thus treating call-transmitted dependences 
fully precisely and heap-transmitted dependences approximately. This approach is known as
CSFI (context-sensitive, field-insensitive) CFL-reachability. 
Consider the CR and PG context-free grammars in~\figref{fig:cfgs}. CSFI amounts to PG-reachability, 
i.e., if the path from \x to \y is in the language L(PG), then \y is PG-reachable from \x. Analogously, 
CIFS amounts to CR-reachability, i.e., if the path from \x to \y is in L(CR), then \y is CR-reachable 
from \x. There are many approximations in the literature. 


\subsection{Imprecision in $G_{\sc BI}$}
\label{sec:imprecision}

Bidirectionality of $G_{\sc BI}$ causes imprecision. Consider the example:

\begin{figure}[t]
\begin{center}
\[
\begin{array}{cc}
\begin{array}{lll}
R & ::= & )_i \;\; | \;\; )_i \; M \;\; | \;\; )_i \; R \;\; | \;\; M \; R \\
C & ::= & (_i \;\; | \;\; (_i \; M \;\; | \;\; (_i \; C \;\; | \;\; M \; C \\
M & ::= & {\sf d} \;\; | \;\; (_i \; M \; )_i \;\; | \;\; {\sf d} \; M \;\; | \;\; (_i \; M \; )_i \; M \\
\end{array}
&

\begin{array}{lll}
G & ::= & r_\f \;\; | \;\; r_\f \; B \;\; | \;\; r_\f \; G \;\; | \;\; B \; G \\
P & ::= & w_\f \;\; | \;\; w_\f \; B \;\; | \;\; w_\f \; P \;\; | \;\; B \; P \\
B & ::= & {\sf d} \;\; | \;\; w_\f \; B \; r_\f \;\; | \;\; {\sf d} \; B \;\; | \;\; w_\f \; B \; w_\f \; B \\
\end{array}
\\
& \\
\mbox{(a) Call-Return (CR) CFG} & \mbox{(b) Put-Get (PG) CFG}

\end{array}
\]
\end{center}
\caption{The CR grammar in (a) captures well-formed call-transmitted paths. 
$R$ captures strings with outstanding (R)eturn edges, e.g., $(_1 \; {\sf d} \; )_1 \; )_2$. 
$C$ captures strings with outstanding (C)all edges, 
e.g., $(_1$, and $M$ captures same-level paths, i.e., paths with matching call and return edges,
e.g., $(_1 \; {\sf d} \; )_1$.
The PG grammar in (b) captures heap-transmitted paths. $P$ captures strings with
outstanding [G]etfield (i.e., field read) edges, e.g., $w_\f \; {\sf d} \; r_\f\; r_g$. 
$P$ captures strings with outstanding [P]utfield (i.e., field write) edges, e.g., $w_\f \; {\sf d}$. 
Finally, $B$ captures strings with matching field write and field read edges, 
e.g., $w_\f \; {\sf d} \; r_\f$.}
\label{fig:cfgs}
\end{figure}


\begin{center}
\begin{tabular}{lll}
& & \\
\begin{minipage}{4cm}
\begin{lstlisting}
if (c) {
  x = a;
}  
else {
  x = b;
}
\end{lstlisting}
\end{minipage}

& 

$\Rightarrow$

&

\begin{tikzcd}[cramped, sep=large]
& {\sf a}  \arrow[dl]  \\
\x \arrow[ur, bend right, dashed] \arrow[dr, bend left, dashed] & \\
& {\sf b} \arrow[ul] 
\end{tikzcd} \\


& &

\end{tabular}
\end{center}

The two forward edges have corresponding inverse edges, creating paths from \x to {\sf b} and from {\sf a} to {\sf b}. 
If {\sf b} is negative, then both {\sf x} and {\sf a} spuriously become negative. The imprecision propagates, ``polluting''
variables throughout the program. 

\cite{Zhang:2017POPL} report that linear conjunctive reachability over $G_{\sc BI}$, a novel highly-precise 
CFL-reachability technique, achieves only modest improvement compared to CSFI over $G_{\sc BI}$; 
they conjecture that this is due to the bidirectionality of $G_{\sc BI}$. 
In contrast, the technique achieves substantial precision improvement over $G_{\sc RI}$~\cite{Zhang:2017POPL}. 
(Recall that $G_{\sc RI}$ is the strict subgraph of $G_{\sc BI}$ that avoids certain inverse edges; we describe  
reference immutability and the construction of $G_{\sc RI}$ in~\secref{sec:referenceImmutability} and~\secref{sec:referenceImmutabilityGraph}.) 
Similarly, \cite{Milanova:2013FTfJP} report substantial negative impact of bidirectionality on 
taint analysis---a taint analysis that removes infeasible edges based on reference immutability information infers 
20\% to 79\% fewer negative variables compared to a taint analysis on the bidirectional graph. (Recall that
the goal of taint analysis is to propagate negative sinks to as few program variables as possible.)

We conducted experiments using the implementation and benchmarks of DroidInfer~\cite{Huang:2015ISSTA} that are made publicly 
available with the artifact of DroidInfer.\footnote{The artifact is publicly available at \url{https://www.cs.rpi.edu/~dongy6/issta-artifact-2015/issta-artifact-2015.zip}; 
the code is available at \url{https://github.com/proganalysis/type-inference.}} 
We ran the analysis on 77 Android apps from the artifact, excluding apps that crashed and apps that contained 0 sources or 0 sinks. 
We ran the taint analysis over $G_{\sc BI}$ and over $G_{\sc RI}$.
\emph{Analysis over $G_{\sc RI}$ reduced the number of reported errors by 41\% on average per app compared to $G_{\sc BI}$. 
An error essentially corresponds to a \emph{(source, sink)} pair, or in other words, to a report of flow from \emph{source} to \emph{sink}; 
the analysis proved 120\% more apps safe compared to analysis over $G_{\sc BI}$}. These results confirm that removing infeasible 
inverse edges benefits analysis precision. 

Inverse edges capture bidirectionality of aliasing. Suppose reference $\x$ flows to $\x'$. Then for all fields \f, $\x.\f$ and $\x'.\f$ are aliases. 
If there is a \emph{write} into $\x'.\f$ then we should be able to \emph{read} that value out of $\x.\f$ and the inverse path from $\x'$ to $\x$ enables that. 
However, if $\x'$ is an immutable reference, then there is no need for the inverse path from $\x'$ to $\x$. 

\subsection{Reference Immutability}
\label{sec:referenceImmutability}


A key goal of this work is to formalize the notion of the inverse edge and understand the role of reference immutability in removing infeasible inverse edges.
Reference immutability~\cite{Tschantz:2005OOPSLA,Huang:2012OOPSLA,Milanova:2018ECOOP} 
ensures that a \emph{readonly} (also called immutable) reference cannot be used to mutate the state of the object it refers to, 
including its transitive state. For example, \x is not \readonly in {\sf y = x; y.f = z;} because it is used in a way that leads to a mutation 
of the object it references; similarly \x is not \readonly in {\sf y.f = x; z = y; w = z.f; w.g = 10;}.

Reference immutability semantics is typically described in terms of a type system, most notably Javari~\cite{Tschantz:2005OOPSLA} and ReIm~\cite{Huang:2012OOPSLA}. Recent work has shown that it can be described in terms of CFL-reachability as well~\cite{Milanova:2018ECOOP}. In this paper, we largely follow the CFL-reachability interpretation of ReIm given in~\cite{Milanova:2018ECOOP}.

In the style of~\secref{sec:bidirectionalFlowGraph}, we analyze each program statement and build a \emph{reference immutability graph} $G$
then decide immutability/mutability of references based on reachability over $G$.  Informally, \x is mutable if it reaches a variable that is updated, such as \y and {\sf w} above.
An assignment statement contributes the following forward edge to $G$. Notably, there are no inverse edges in reference immutability:
\[\mbox{ASSIGN}\quad \llbracket \code{x = y}\rrbracket(G) = G\cup \{ \y \arrow{d} \x \} \]

Similarly to~\secref{sec:bidirectionalFlowGraph} a method call creates call and return forward edges as expected:
\[\mbox{CALL/RET} \quad \llbracket i\!: \code{x = y.m(z)}\rrbracket(G) = G\cup \{\y \callarrow{i} \this \} \cup \{ \z \callarrow{i} \p \} \cup \{ \ret \retarrow{i} \x \} \]

A difference with~\secref{sec:bidirectionalFlowGraph} arises in the handling of heap-transmitted dependences. 
Together, a pair of field write \code{x.f = y} and field read $\code{y' = x'.f}$ contribute the following edges
\[\mbox{WRITE/READ} \quad \llbracket \code{x.f = y}, \; \code{y' = x'.f} \rrbracket (G) = G \cup \{ \y \arrow{d} \x.\f \approxarrow \x'.\f \arrow{d} \y' \} \cup \{ \x' \arrow{d} \x'.\f \}\]
The first set of edges creates a path from \y to $\y'$. Here $\x.\f \approxarrow \x'.\f$ is an \emph{approximate edge}.
In terms of Reps' terminology~\cite{Reps:2000TOPLAS}, the reference immutability semantics models heap-transmitted 
dependences approximately. The approximation comes from the fact that 
regardless of whether $\x.\f$ and $\x'.\f$ are aliases, the semantics propagates $\y'$ back to $\y$. 
If there is an update of $\y'$, e.g., $\y'.\g = 5$, \y will be determined mutable. This is precisely what
makes inverse edges unnecessary here. A field read contributes an additional edge $\x' \arrow{d} \x'.\f$
needed to capture mutation to transitive state.
If there is a mutation of $\y'$, this last edge propagates the mutation back to $\x'$.

An \emph{update} is a node \y such that there is a write statement of the form $\code{y.f = z}$.
Essentially, we are interested if there is an ``appropriately annotated'' path in $G$ from a reference \x to an update node. For example, 
consider the code and its corresponding graph $G$. {\sf id} is the standard identity function:
\[i: \code{x = id(y); z = x; w = z.f; w.g = 10;} \quad \Rightarrow \quad \y \callarrow{i} \p \arrow{d}  \ret \retarrow{i} \x \arrow{d} \z \arrow{d} \z.\f \arrow{d} {\sf w} \]
There is a path from \y to the update {\sf w} with properly matched parenthesis which means that {\sf y} is a mutable reference; 
the object $o$ that \y refers to is modified through \y as the assignment of \y to parameter \p of {\sf id}
leads to the mutation $o.\f.\g = 10$.

A \emph{call-transmitted} path contains no approximate edges and it is well-formed in CR, i.e., its annotations form a string in 
L(CR) (recall~\figref{fig:cfgs}(a)). 
A \emph{heap-transmitted} path is made up of call-transmitted paths interleaved with approximate edges. 
Let $U$ denote the set of all call-transmitted and heap-transmitted paths to updates in $G$.
We break paths in $U$ into 2 categories:
(1) \emph{$M/C$-paths} and (2) \emph{$R$-paths}. A path in $U$ is an $M/C$-path
if and only if the annotations on the \emph{leading}, i.e., first, call-transmitted path 
form a string in the language described by $M$ (i.e., calls and returns balance out), or they form a string in 
the language described by $C$ (i.e., outstanding calls).
For example, ${\sf e} \callarrow{6} \this_{\sf set}$ is a $C$-path.
A path in $U$ is an $R$-path if and only if the edge annotations on the leading call-transmitted path 
form a string in $R$, e.g., $\ret \retarrow{7} {\sf b}$ is an $R$-path. 
We consider that there is a trivial path from each node to itself, so an update node is \mutable.

\begin{itemize}
\item[(1)] Examples of $M/C$-paths, following the graph in~\secref{sec:overviewCFL}, include (assuming {\sf b} is an update):
\[{\sf e} \callarrow{7} \this_{\sf get} \arrow{d} \this_{\sf get}.\f \arrow{d} \ret \retarrow{7} {\sf \underline{b}} \mbox{     (leading call-transmitted path is an $M$-path)} \]
and
\[{\sf a} \callarrow{6} \p \arrow{d} \this_{\sf set}.\f \approxarrow \this_{\sf get}.\f \arrow{d} \ret \retarrow{7} {\sf \underline{b}} \mbox{     (leading call-transmitted path is a $C$-path)} \]
\item[(2)] examples of $R$-paths include (again assuming {\sf b} is an update):
\[\this_{\sf get} \arrow{d} \this_{\sf get}.\f \arrow{d} \ret \retarrow{7} {\sf \underline{b}} \]
\end{itemize}

Reference immutability~\cite{Tschantz:2005OOPSLA, Huang:2012OOPSLA, Milanova:2018ECOOP} classifies variables as follows:
\[
\begin{array}{ll}
\x \mbox{ is } \mutable & \mbox{if there is an $M/C$-path from \x to an update} \\
\x \mbox{ is } \poly & \mbox{if there is no $M/C$-path but there is an $R$-paths from \x to an update} \\
\x \mbox{ is } \readonly & \mbox{if there is neither $M/C$-path nor $R$-path from \x to an update}
\end{array} 
\]
In the above examples {\sf e} and {\sf a} are both \mutable and $\this_{\sf get}$ and \ret are \poly, due to the $R$-paths to the update. 
If we removed the assumption that {\sf b} is an update, then {\sf e}, {\sf a}, $\this_{\sf get}$ and \ret above would be \readonly. 
Intuitively, an $M/C$-path from \x unequivocally makes \x mutable---mutation is immediate or within the immediate call.
An $R$ path does not necessarily make \x mutable; it is mutable in the context of the $R$-path, but it may be readonly in the context of other paths. 
For example, {\sf \mutable b = e.get()} makes $\this_{\sf get}$ \mutable in the context of this path (precisely the $R$-path above), 
however, {\sf \readonly d = g.get()} leaves $\this_{\sf get}$ \readonly in the context of this different path: $\this_{\sf get} \arrow{d} \this_{\sf get}.\f \arrow{d} \ret \retarrow{9} {\sf d}$.

Lastly, we summarize the \emph{adaptation operation}~\cite{Dietl:2011ECOOP,Huang:2012OOPSLA} 
which plays a role in the theorem that proves soundness of CFL-reachability over $G_{\sc RI}$.
Viewpoint adaptation, written as $\x \rhd \p$, interprets the immutability type of \p in the context of \x.
We simplify the adaptation notation to work on variables rather than immutability qualifiers. $\x \rhd \p$ refers to the 
immutability types of \x and \p, not to the variables themselves. 
\[
\begin{array}{rllll}
  \x &\rhd& \readonly &=& \readonly \\
  \x &\rhd& \mutable &=&  \mutable \\
  \x  &\rhd& \poly &=&  \x \\
\end{array}
\]
 A \readonly or \mutable{} \p remains as is, regardless of the context \x. A \poly{} \p takes the type of \x, which makes
adaptation interesting. In reference immutability, left-hand sides \x of call assignments serve as contexts 
 of adaptation. This is because, intuitively, there are multiple paths from a \poly variable \p in \m, one through each one of the left-hand sides of calls to \m;
 the mutability status of each path is determined by the left-hand side of the call.
 If $i: \; {\sf x = m(z)}$ is such that \x is \mutable, then \p is mutable at $i$. 
 Returning to the examples above, at $\mutable \; {\sf b = e.get()}$ we have ${\sf b} \rhd \this_{\sf get} = \mutable$ and ${\sf b} \rhd \ret = \mutable$ reflecting 
 that $\this_{\sf get}$ and \ret are \mutable at the call to {\sf get} at 7 because the left-hand side, i.e., the context {\sf b} is \mutable. 
 In contrast, at $\readonly \; {\sf d = g.get()}$ we have ${\sf d} \rhd \this_{\sf get} = \readonly$ 
 and ${\sf d} \rhd \ret = \readonly$.~\footnote{Standard viewpoint adaptation~\cite{Dietl:2011ECOOP} 
 uses the receiver as context of adaptation, e.g., in \code{x = y.m(z)}, the context of adaptation is \y. The use of left-hand side \x is non-standard 
 and reflects the specific semantics of reference immutability; ~\citep{Huang:2012OOPSLA} elaborates on this.}
 
 We generalize adaptation to a sequence of contexts, 
 e.g., $\x_i \rhd \x_j \rhd \x_k \rhd \ret$ adapts \ret in a larger context. Suppose $\ret$ is \poly, and left-hand-sides $\x_k$ at call $k$ and
 $\x_j$ at call $j$ are \poly, however, $\x_i$ at the outermost call site $i$ is \readonly; \ret in context $\x_i \rhd \x_j \rhd \x_k$ is \readonly. 

\begin{center}
\begin{minipage}{5.5cm}
\begin{lstlisting}
A id0(A p0) { ret0 = p0; }

A id1(A p1) { ret1 = id0(p1); }

A id2(A p2) { ret2 = id1(p2); }
...
mutable b = id2(a); 
b.f = z;
readonly d = id2(c);
\end{lstlisting}
\end{minipage}
\end{center}
In the above example, {\sf ret0}, {\sf p0}, {\sf ret1}, {\sf p1}, {\sf ret2}, and {\sf p2} are all \poly due to the $R$-paths to update {\sf b}. In context 7 {\sf ret0}  
is interpreted as ${\sf b} \rhd {\sf ret2} \rhd {\sf ret1} \rhd {\sf ret0} = \mutable \rhd \poly \rhd \poly \rhd \poly = \mutable$. ${\sf b} \rhd {\sf ret2} \rhd {\sf ret1}$
abstracts the stack context of {\sf id0} that line 7 initiates. 
 
\subsection{Flow Graph $G_{\sc RI}$}
\label{sec:referenceImmutabilityGraph}



We use reference immutability to build a new flow graph $G_{\sc RI}$ without certain infeasible inverse edges.
In summary, explicit and implicit assignments forgo inverse edges if the left-hand-side of the assignment is \readonly
as illustrated by the rule for assignment statement. 
\[\mbox{ASSIGN}\quad \llbracket \code{x = y} \rrbracket(G_{\sc RI}) = \left\{ 
\begin{array}{ll}
G_{\sc RI} \cup \{ \y \arrow{d} \x \} \cup \{ \x \inversearrow{d} \y \} & \mbox{ if \x is not \readonly } \\
G_{\sc RI} \cup \{ \y \arrow{d} \x \} & \mbox{ otherwise } \\
\end{array}
\right.
\]
The rules for WRITE, READ, CALL and RET are analogous. CALL $i\!: \code{x = y.m(z)}$ adds inverse edge $\this \inverseretarrow{i} \y$ when $\x \rhd \this \neq \readonly$, 
and it adds $\p \inverseretarrow{i} \z$ when $\x \rhd \p \neq \readonly$. RET adds $\x \inversecallarrow{i} \ret$ when $\x \rhd \ret \neq \readonly$.

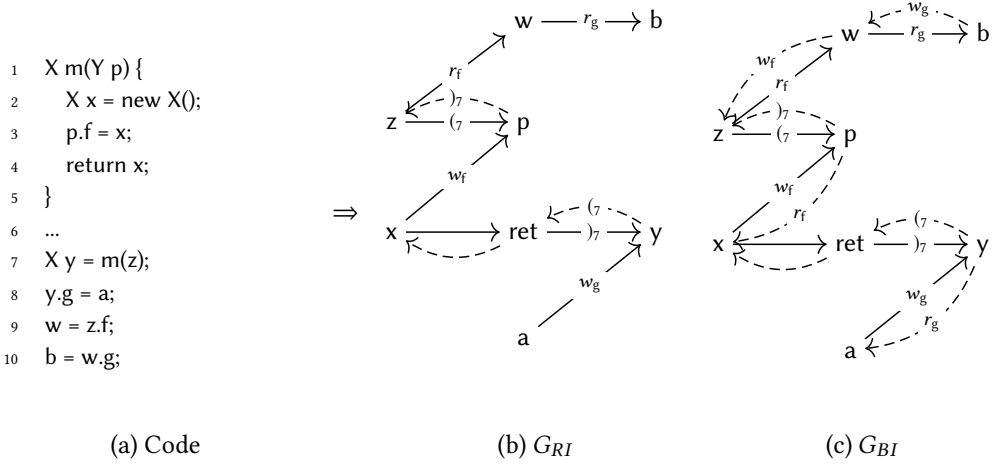
\begin{figure}
\begin{tabular}{cccc}
\begin{minipage}{4cm}
\begin{lstlisting}
X m(Y p) { 
  X x = new X(); 
  p.f = x; 
  return x; 
}
...
X y = m(z);
y.g = a;
w = z.f;
b = w.g;
\end{lstlisting}
\end{minipage}

&
$\Rightarrow$ 

& 

\begin{minipage}{4cm}
\begin{tikzcd}[cramped, sep=large]
   & \w \arrow[r,"{r_\g}" description] & {\sf b} \\ 
\z \arrow[ur,"{r_\f}" description]  
    \arrow[r, "{(_7}" description] & \p \arrow[l, dashed, bend right, "{)_7}" description] & \\       
\x \arrow[ur,"{w_\f}" description] 
    \arrow[r] & \ret \arrow[r, "{)_7}" description] 
                            \arrow[l, dashed, bend left] & \y \arrow[l, dashed, bend right, "{(_7}" description ]  \\
& {\sf a} \arrow[ur, "{w_\g}" description] & \\
& & &
\end{tikzcd} 
\end{minipage}

&

\begin{minipage}{4cm}
\begin{tikzcd}[cramped, sep=large]
   & \w \arrow[r, "{r_\g}" description ] \arrow[dl, dashed, bend right, "{w_\f}" description] & {\sf b} \arrow[l, dashed, bend right, "{w_\g}" description] \\ 
\z \arrow[ur,"{r_\f}" description]  
    \arrow[r, "{(_7}" description] & \p \arrow[l, dashed, bend right, "{)_7}" description]
                                      \arrow[dl, dashed, bend left, "{r_\f}" description] &  \\       
\x \arrow[ur,"{w_\f}" description] 
    \arrow[r] & \ret \arrow[r, "{)_7}" description] 
                            \arrow[l, dashed, bend left] & \y \arrow[l, dashed, bend right, "{(_7}" description ] 
                                                                               \arrow[dl, dashed, bend left, "{r_\g}" description ] \\
& {\sf a} \arrow[ur, "{w_\g}" description] & \\
& & &
& & &
\end{tikzcd} 
\end{minipage} \\

(a) Code & & (b) $G_{\sc RI}$ & (c) $G_{\sc BI}$ 

\end{tabular}
\caption{An example contrasting $G_{\sc RI}$ and $G_{\sc BI}$. $G_{\sc RI}$ retains the path from {\sf a} to {\sf b} but avoids multiple infeasible paths from $G_{\sc BI}$, e.g., the path fro {\sf b} to {\sf a}.}
\label{fig:example3}
\end{figure}

Consider the example in \figref{fig:example3}. 
\p and \y are updates and therefore \mutable, and \z is \mutable as well due 
to the $C$-path to update \p: $\z \callarrow{7} \p$. \x and \ret are \poly due to the $R$-path $\x \arrow{d} \ret \retarrow{7} \y$.
However, \w, {\sf a}, {\sf b}, as well as fields \f and \g are \readonly, as there is neither $M/C$-path
nor $R$-path to an update. The rules above entail only 3 inverse edges: (1) $\p \inverseretarrow{7} \z$, 
(2) $\y \inversecallarrow{7} \ret$, and (3) $\ret \inversearrow{d} \x$, precisely the
edges needed to capture the path from {\sf a} to {\sf b}. 

Our key result is that paths in $G_{\sc RI}$ with properly matched $w/r$ and call/ret annotations
capture all chains as defined in~\secref{sec:dynamicSemantics}.

\section{Soundness of CFL-reachability over $G_{\sc RI}$}
\label{sec:soundness}


To prove soundness of reachability over $G_{\sc RI}$, we first define the abstraction function 
$\alpha(A)$ over stack contexts $A$: \[\alpha(\langle {\sf main}, f_1, f_2 ... f_n \rangle) =  \langle i_1,i_2 ... i_n \rangle \]
where $i_1$, $i_2$, ... $i_n$ are the static call sites that triggered frames $f_1$, $f_2$, ... $f_n$. $\alpha(A)$ 
extends to partial contexts $A$ as follows: $\alpha(\langle f_2 ... f_n \rangle) =  \langle i_2 ... i_n \rangle$.
\[\alpha(B) \rhd \y = \alpha(\langle {\sf main}, f_1, f_2 ... f_n \rangle) \rhd \y = \x_{i_1} \rhd \x_{i_2} \rhd ... \rhd \x_{i_n} \rhd \y \]
where $\x_{i_1}$, $\x_{i_2}$, ... $\x_{i_n}$ are the left-hand sides of call assignments $i_1$, $i_2$, ... $i_n$. 

We define the abstract difference, $\alpha(\Delta AB)$, as the tuple $(\alpha(A-D), \alpha(B-D))$ where $D$ is the longest common prefix of $A$ and $B$. 
We will denote these tuples as $(\mathit{ret}, \mathit{call})$ as $\alpha(A-D)$ abstracts a certain \emph{return sequence} and 
$\alpha(B-D)$ abstracts a certain \emph{call sequence} as we discussed in~\secref{sec:frame}. 
We define the concatenation operation over abstract differences as follows: 
\[ (\mathit{ret}_1, \mathit{call}_1) \oplus (\mathit{ret}_2, \mathit{call}_2) = \left\{ 
\begin{array}{ll}
(\mathit{ret}_1, (\mathit{call}_1 \! - \! \mathit{ret}_2)\!+\!\mathit{call}_2) & \mbox{if } \mathit{ret}_2 \mbox{ is a suffix of } \mathit{call}_1 \\
(\mathit{ret}_1 \! + \! (\mathit{ret}_2\!-\!\mathit{call}_1), \mathit{call}_2) & \mbox{if } \mathit{call}_1 \mbox{ is a suffix of } \mathit{ret}_2 \\
\mathit{undefined} & \mbox{otherwise} 
\end{array}
\right.
\]
In the above, we overload \emph{minus} $-$ to subtract a suffix, not just a prefix of a string and $+$ is just standard string concatenation.
Consider $\alpha(\Delta AB) = (\alpha(A-D), \alpha(B-D)) = (\mathit{ret}_1, \mathit{call}_1)$ and $\alpha(\Delta BC) = (\alpha(B-E), \alpha(C-E)) = (\mathit{ret}_2, \mathit{call}_2)$. 
If we think of $\mathit{ret}$ strings as strings of closing parentheses, i.e., $\mathit{ret} = \langle i_1,i_2 ... i_n \rangle = \; )_{i_n} ... )_{i_2} \; )_{i_1}$, 
and of $\mathit{call}$ strings as strings of opening parentheses, i.e., $\mathit{call} =  \langle j_1,j_2 ... j_m \rangle = \; (_{j_1} \; (_{j_2} ... (_{j_m}$, then concatenation
cancels out call annotations $(_k$ and return annotations $)_k$. For the remainder of this section we will interpret
$\alpha(\Delta AB)$ to mean both the tuple $(\mathit{ret},\mathit{call})$ as defined above and the string $)_{i_n} ... )_{i_2} \; )_{i_1} (_{j_1} \; (_{j_2} ... (_{j_m}$.
In the abstract domain, graph $G_{\sc RI}$, we will be looking at paths from \x to \y; a chain $(\x^A, \y^B)$ will map into a path from \x to \y in $G_{\sc RI}$ 
such that the string of unmatched call and return annotations on that path is exactly $\alpha(\Delta AB)$.
If $\mathit{call}_1 = \alpha(B-D)$ is longer than $\mathit{ret}_2 = \alpha(B-E)$, 
then $\alpha(B-D)$ cancels out $\alpha(B-E)$ and the outstanding call annotations are prepended onto $\alpha(C-E)$. 
Analogously, if $\alpha(B-E)$ is longer than $\alpha(B-D)$ then $\alpha(B-E)$ cancels out $\alpha(B-D)$, and
the outstanding return annotations are appended to $\mathit{ret}_1=\alpha(A-D)$. If neither $\mathit{call}_1$ cancels out $\mathit{ret}_2$, 
nor $\mathit{ret}_2$ cancels $\mathit{call}_1$, then concatenation is undefined because the strings represent distinct runtime contexts.

As an example, return to~\figref{fig:runningExample}. $\p^{\langle {\sf main}, f_1 \rangle}$ flows to $\ret^{\langle {\sf main}, f_3 \rangle}$, 
and we are interested in the abstract difference $\alpha(\Delta \langle {\sf main}, f_1 \rangle \langle {\sf main}, f_3 \rangle)$ which is $(\langle 6\rangle, \langle 7\rangle)$ or just the string $)_6 (_7$. Similarly, since $\ret^{\langle {\sf main}, f_3 \rangle}$ flows to ${\sf b}^{\langle {\sf main}\rangle}$ we are interested in $\alpha(\Delta \langle {\sf main}, f_3 \rangle \langle {\sf main} \rangle)$
which is just $)_7$, with an empty call sequence. $)_6 (_7 \; \oplus \; )_7$ equals $)_6$. 

The concatenation lemma below states that if we have the strings $\alpha(\Delta AB)$ and $\alpha(\Delta BC)$ their concatenation as defined above produces 
$\alpha(\Delta AC)$, precisely the abstraction of $\Delta AC$.

\begin{lemma}
$\alpha(\Delta AB) \oplus \alpha(\Delta BC) = \alpha(\Delta AC)$
\label{abstractDifferenceLemma}
\end{lemma}
\begin{proof}
Let $\alpha(\Delta AB) = (\alpha(A-D), \alpha(B-D)) = (\mathit{ret}_1, \mathit{call}_1)$ and let $\alpha(\Delta BC) = (\alpha(B-E), \alpha(C-E)) = (\mathit{ret}_2, \mathit{call}_2)$.
Here $D$ is the longest common prefix of $A$ and $B$ and $E$ is the longest common prefix of $B$ and $C$. 
There are two cases, (1) $D \le E$ and (2) $E \le D$. We argue case (1), case (2) is analogous. Since $D \le E$, it follows that $\mathit{ret}_2 = \alpha(B-E)$ is a suffix of 
$\mathit{call}_1 = \alpha(B-D)$. By the definition of concatenation above we have
\[\alpha(\Delta AB)\oplus \alpha(\Delta BC) = (\alpha(A-D), \alpha(C-E)\!+\!\alpha(E-D)) = (\alpha(A-D), \alpha(C-D))\] 
It remains to make the argument that $D$ is the longest common prefix of $A$ and $C$ which is true because $E-D$ does not overlap with $A-D$ or otherwise the longest common prefix of $A$ and $B$ would have been longer than $D$. Therefore $(\alpha(A-D), \alpha(C-D)) = \alpha(\Delta AC)$ which establishes the statement.
\end{proof}

Our main theorem, Theorem~\ref{theTheorem}, shows that every chain $(\x^A, \y^B)$ in $\mathbb{C}$ is represented by an appropriately annotated path from \x to \y in $G_{\sc RI}$.
We write $\x \patharrow{\alpha(\Delta A B)} \y$ to denote the existence of a path from $\x$ to $\y$ in $G_{\sc RI}$ with a string $s \in $ L(PG) $\cap$ L(CR)\footnote{As it is standard, for the purposes of the intersection, we extend the alphabets of grammars PG and CR to include the symbols of the other grammar.}, where the PG (i.e., $w$ and $r$) component of $s$ is balanced, and the CR (i.e., $($ and $)$) component of $s$ contains exactly the $\alpha(A-D)$ string of unbalanced returns, followed by the $\alpha(B-D)$ string of unbalanced calls. ($D$ is the longest common prefix as expected.)  As an example, 
consider $\x \stackrel{w_\f}{\longrightarrow} \p \retarrow{7} \z  \stackrel{r_\f}{\longrightarrow} \w$ in~\figref{fig:example3}(b);
the field component is balanced while the call/ret component reflects that \x flows from the context of the call at line 7 back into {\sf main}.
It is easy to see that if there is a path $\x \patharrow{\alpha(\Delta A B)} \y$ and a path $\y \patharrow{\alpha(\Delta B C)} \z$, 
then there is a path $\x \patharrow{\alpha(\Delta A B)\oplus \alpha(\Delta B C)} \z$.
By the concatenation lemma, this is exactly $\x \patharrow{\alpha(\Delta A C)} \z$.



\begin{theorem}
\label{theTheorem}
Let $\mathbb{C},\mathbb{S},\mathbb{H}$ be a program state and let $(\x^A, \y^B) \in \mathbb{C}$. The following statements are true.
\begin{itemize}
\item There is a path $\x \patharrow{\alpha(\Delta A B)} \y$ in $G_{\sc RI}$.
\item If $\alpha(B) \rhd \y$ is \mutable according to reference immutability (ReIm), then there is an inverse path $\y \patharrow{\alpha(\Delta B A)} \x$ in $G_{\sc RI}$.
\end{itemize}
\end{theorem}

A corollary is that if \y is an update, i.e., we have a write \code{y.f = z}, then there is an inverse path $\y \patharrow{\alpha(\Delta B A)} \x$ in $G_{\sc RI}$. 
(By definition ReIm types an update as \mutable and therefore $\alpha(B) \rhd \y$ is \mutable; the inverse path follows from the second clause of the theorem.)
An inverse path is not necessarily made up of inverse edges, it is just the ``inverse'' of $\x \patharrow{\alpha(\Delta A B)} \y$. 

\begin{proof} The prove is by standard induction and case-by-case analysis of program statements. Given transition $\llbracket s \rrbracket(A,\mathbb{C},\mathbb{S},\mathbb{H}) = \mathbb{C'},\mathbb{S'},\mathbb{H'}$ 
if the lemma holds on \emph{all states} up to $\mathbb{C},\mathbb{S},\mathbb{H}$, including $\mathbb{C},\mathbb{S},\mathbb{H}$, then it holds on $\mathbb{C'},\mathbb{S'},\mathbb{H'}$. We note that even though 
the structure of the proof is standard, the treatment of individual statements is involved. 


Case 1. Consider statement $\x = \y$. By the inductive hypothesis, there is a path $\vv \patharrow{\alpha(\Delta A B)} \y \in G_{\sc RI}$ for every 
$(\vv^A,\y^B) \in \mathbb{C}$. Since there is an edge $\y \rightarrow \x \in G_{\sc RI}$ (by construction of $G_{\sc RI}$), there is a path 
$\vv \patharrow{\alpha(\Delta A B)} \x \in G_{\sc RI}$.
To show the second clause of the theorem, assume $\alpha(B) \rhd \x$ is \mutable. Therefore, $\alpha(B) \rhd \y$ is \mutable, and by the inductive
hypothesis there is an inverse path $\y \patharrow{\alpha(\Delta B A)} \vv \in G_{\sc RI}$. $\alpha(B) \rhd \x = \mutable$ implies that \x is either 
\poly (i.e., there is a $R$-path to an update), or \mutable (an $M/C$-path to update). Thus, there is an inverse edge $\x \rightarrow \y \in G_{\sc RI}$ by construction. 
Adding the inverse edge to the inverse path yields $\x \patharrow{\alpha(\Delta B A)} \vv$, as needed (concatenation lemma).

Case 2. Consider call $i: \; \x = \y.\m(\z)$. By induction, there is path $\vv \patharrow{\alpha(\Delta A B)} \z \in G_{\sc RI}$ for each 
$(\vv^A,\z^B) \in \mathbb{C}$. The dynamic semantics appends the new fresh frame $f$ onto $B$ to get new context $B' = B\! \oplus \! f$, 
and new chain $(\vv^A,\p^{B'}) \in \mathbb{C}'$. There is an edge $\z \callarrow{i} \p$ and thus path $\vv \patharrow{\alpha(\Delta A B')} \p \in G_{\sc RI}$. 
Again, the second clause is more involved. Assume $\alpha(B') \rhd \p$ is \mutable. Then $\alpha(B) \rhd \z$ is \mutable as well. 
(If \p is \mutable, then by the rules of ReIm, \z is \mutable, and therefore, $\alpha(B) \rhd \z$ is mutable. 
Otherwise, \p is \poly and $\alpha(B')$ is \mutable. Since $\alpha(B') = \alpha(B) \rhd \x$, there are two cases, $\alpha(B)$ is \mutable 
and \x is \poly, or \x is \mutable.
In the first case, by the rules of ReIm we must have \z \poly, which yields $\alpha(B) \rhd \z = \mutable$. If \x is mutable, then 
by the rules of ReIm \z must be \mutable, which again yields $\alpha(B) \rhd \z = \mutable$.) Since we have established that 
$\alpha(B) \rhd \z$ is mutable, we have an inverse path $\z \patharrow{\alpha(\Delta B A)} \vv \in G_{\sc RI}$.
Since $\alpha(B') \rhd \p$ is \mutable, \p is either \mutable or \poly. In the case of \poly, we must have $\alpha(B') = \mutable$ which 
entails that \x is not \readonly, and therefore, $\x \rhd \p$ is not \readonly; thus, the analysis adds an inverse edge $\p \inverseretarrow{i} \z$.
Therefore, $\p \inverseretarrow{i} \z \patharrow{\alpha(\Delta B A)} \vv \Rightarrow \p \patharrow{\alpha(\Delta B' B)} \z \patharrow{\alpha(\Delta B A)} \vv \Rightarrow \p \patharrow{\alpha(\Delta B' A)} \vv$ (by Lemma~\ref{abstractDifferenceLemma}), which is the expected inverse path.

Case 3, return $i: \; \x = \y.\m(\z)$ is analogous to Case 2. 

Case 4. The most interesting case arises when the runtime semantics adds new chains at a field read. Let $\y = \x.\f$ in context $B$ and 
$\x'.\f = \y'$ in $A$ be such that \x and $\x'$ refer to $o$ and $\x'.\f = \y'$ is the most recent write to $o.\f$ preceding the read out of $\x.\f$. That is, this
is the last write that set $[o.\f \leftarrow ...]$ to form some $\mathbb{C}''$. By the inductive hypothesis,
for every chain $(\vv^D,(\y')^A)$, in this case in the earlier $\mathbb{C}''$, there is a path $\vv \patharrow{\alpha(\Delta D A)} \y' \in G_{\sc RI}$. By Lemma~\ref{lemma:chainLemma} there 
are chains $(\w^C,(\x')^A) \in \mathbb{C}''$ and $(\w^C,\x^B) \in \mathbb{C}$, where $\w = \mathit{new}\; o$ in context $C$ is the allocation site of $o$. 
By the inductive hypothesis, we have paths $\w \patharrow{\alpha(\Delta C A)} \x' \in G_{\sc RI}$, $\w \patharrow{\alpha(\Delta C B)} \x \in G_{\sc RI}$, as
well as an inverse path $\x' \patharrow{\alpha(\Delta A C)} \w \in G_{\sc RI}$ since $\x'$ is \mutable. Thus, we have a path
\[ \vv \patharrow{\alpha(\Delta D A)} \y' \stackrel{w_\f}{\rightarrow} \x' \patharrow{\alpha(\Delta A C)} \w \patharrow{\alpha(\Delta C B)} \x \stackrel{r_\f}{\rightarrow} \y 
\quad \Rightarrow \quad \vv \patharrow{\alpha(\Delta D B)} \y \in G_{\sc RI} \]

Now consider the second clause of the theorem. Suppose $\alpha(B) \rhd \y$ is \mutable; we need to show inverse path $\y \patharrow{\alpha(\Delta B D)} \vv \in G_{\sc RI}$. If $\alpha(B) \rhd \y$ is \mutable, then \y is \mutable or \poly and by the rules of ReIm \f is \poly and $\x'$ is \mutable. 
Therefore $\y'$ is \mutable and there is an inverse path 
$\y' \patharrow{\alpha(\Delta A D)} \vv \in G_{\sc RI}$. Next, if $\alpha(B) \rhd \y$ is \mutable, then $\alpha(B) \rhd \x$ is \mutable and by the inductive
hypothesis we have an inverse $\x \patharrow{\alpha(\Delta B C)} \w \in G_{\sc RI}$. \y being \mutable or \poly also implies that we have added an inverse edge
$\y \stackrel{w_\f}{\rightarrow} \x$ during construction of $G_{\sc RI}$. Adding all these paths along with the original $\w \patharrow{\alpha(\Delta C A)} \x'$ yields
\[ \y \stackrel{w_\f}{\rightarrow} \x \patharrow{\alpha(\Delta B C)} \w \patharrow{\alpha(\Delta C A)} \x' \stackrel{r_\f}{\rightarrow} \y'  \patharrow{\alpha(\Delta A D)} \vv 
\quad \Rightarrow \quad \y \patharrow{\alpha(\Delta B D)} \vv \in G_{\sc RI} \]

\end{proof}

Of course, even though we have shown that each chain has, roughly speaking, a corresponding path $p$ in the intersection of L(PG) and L(CR), computing the exact set of paths $p$ is undecidable. We define an approximate reachability analysis over $G_{\sc RI}$, {\sc CFL}, and we show that each path $p$ has a corresponding representative path in {\sc CFL}. 
In the following section we define type-based FlowCFL and interpret it in terms of CFL-reachability. In~\secref{sec:equivalence} we define the {\sc CFL} algorithm
and establish equivalence of FlowCFL and {\sc CFL} thus proving FlowCFL correct. 

\section{Type-based Analysis}
\label{sec:typeBasedSemantics}

This section presents the type-based FlowCFL outlining parallels with CFL-reachability as described in~\secref{sec:CFL}. 
The reader may wonder why one needs a type system, when one has a clear semantics 
in terms of standard CFL-reachability. First, type systems and type-based taint analysis have 
already been used in the literature~\cite{Sampson:2011PLDI, Huang:2014FASE, Huang:2015ISSTA, Shankar:2001USENIXSecurity}, 
in some cases without correctness proofs. 
CFL-reachability brings insight into type-based reachability/taint analysis and the theory of type qualifiers~\cite{Foster:2002PLDI}, and presents a novel framework 
for reasoning about correctness. 
Second, a type system allows programmers to specify requirements with type qualifiers, 
e.g., \high{} \x, and take advantage of systems such as the Checker Framework 
(\texttt{https://checkerframework.org/}) 
to statically check these requirements; such requirements cannot be easily expressed or checked using CFL-reachability. 
Third, type systems are modular, while CFL-reachability systems are typically whole-program analyses. 
A significant advantage of a type-based interpretation is that it allows for modular reasoning. We can infer type annotations
for libraries (e.g., as in~\cite{Huang:2012OOPSLA}), then type check a user program against annotated libraries while handling 
callbacks via standard function subtyping. 

\secref{sec:typeQualifiers} describes the type qualifiers in FlowCFL and~\secref{sec:typingRules} describes the typing rules. 
\secref{sec:equivalence} establishes equivalence of a certain CFL-reachability analysis and the type-based analysis.




\subsection{Type Qualifiers}
\label{sec:typeQualifiers}

FlowCFL makes use of the \high{} and \low{} type qualifiers that we introduced in~\secref{sec:overview}:
\begin{itemize}
\item \high{} --- a \high variable \x is a source or \x is such that a source flows to \x. A \high{} \x or any of its components cannot flow to a sink.
For example 
\[ \y = \x.\f \]
where \y is a sink, is not allowed. Similarly, 
\[ \code{y = id(x); z = y.f; } \]
where \z is a sink and {\sf id} is the identity function, is not allowed.   
\item \low{} --- a \low variable \x is a sink, or \x flows to a sink.  
\item \poly{} --- a \poly variable expresses polymorphism. In some contexts, \poly is interpreted as \high and in other contexts, it is interpreted as \low.
\end{itemize}

The subtyping hierarchy with \poly becomes
\[
\low <: \poly <: \high
\]
It is allowed to assign a \poly variable into a \high one, but not the other way around; similarly, it is allowed to assign a \low variable
into a \poly one, but not the other way around. Subtyping in FlowCFL models flow of values which is non-standard. The \poly value is interpreted as 
either \low or \high, depending on the stack context. It would be safe to assign a \low value into a \poly variable (which becomes 
either \low or \high) without causing flow from \high to \low. However, it would not be safe to assign a \high value into a \poly variable
because it may become \low, causing flow from a \high to a \low variable.  

We define the adaptation operation, analogously to the operation in~\secref{sec:referenceImmutability}. 
\[
\begin{array}{rllll}
  \_ &\rhd& \high &=& \high \\
  \_ &\rhd& \low &=&  \low \\
  q  &\rhd& \poly &=&  q \\
\end{array}
\]
Again, a \high or \low variable remains \high or \low. As in~\secref{sec:referenceImmutability} a \poly
variable takes the value of the \emph{adapter} (i.e., context of adaptation): if the adapter is \high, 
then \poly is interpreted as \high, and if the adapter is \low then \poly is interpreted as \low.

To avoid clutter we have used the same notation for the adaptation operator, $\rhd$, as in~\secref{sec:referenceImmutability}. 
From now on, we will use $\rhd$ to refer to the above definition (FlowCFL), and we will use $\rhdri$ to refer to the adaptation 
operator of reference immutability in the rare occasions it comes into play.

Adaptation adapts fields, formal parameters, and return values according to the context 
at the field access and method call. The type of a \poly field \code{f} takes the value of the receiver at the field access. 
The type of a \poly parameter or return is interpreted by
adapters at call site $i$. We elaborate on this shortly.


\subsection{Typing Rules}
\label{sec:typingRules}

The typing rules for program statements appear in~\figref{fig:types}.
The rules are defined in terms of a type environment $\Gamma$, which is standard.
$\Gamma = \langle \mathcal{C}, \sigma \rangle$, where $\mathcal{C}$ is a set of subtyping constraints
and $\sigma$ is a map from program variables to type qualifiers:
$\sigma\!: V \rightarrow \{ \high, \poly, \low \}$.
The premise of each rule in~\figref{fig:types} consists of two parts: one part adds constraints to $\mathcal{C}$, 
and the other part, ``$\mathcal{C}$ holds'', enforces those constraints. Concretely, 
``$\mathcal{C}$ holds'' requires (1) that $\mathcal{C}$ is closed under the rules in~\figref{fig:transitive}
and (2) that assignment of qualifiers to variables is such that all subtyping constraints in $\mathcal{C}$ hold.

Rule \rulename{tassign} adds constraint  $\y <:  \x$ to $\mathcal{C}$, which
forbids assignment of a \high or \poly reference to a \low one as well as assignment
of a \high reference to a \poly one. Again, we abuse notation 
by eliding qualifiers. Strictly, the above constraint should have been written as $q_\y <: q_\x$
where $q_\y = \Gamma(\y)$ and $q_\x = \Gamma(\x)$; rules are more compact while still clear.
If \x is not \readonly according to reference immutability, \rulename{tassign}
adds the \emph{inverse} constraint $\x <: \y$ as well. In other words, the expected subtyping constraint
turns into an equality constraint. This is a well-known issue, typically referred to as the problem of 
covariant arrays~\cite{Myers:1997POPL}, which stipulates that subtyping is unsafe in the presence of mutable 
references. The standard solution, adopted by the majority of systems, e.g., EnerJ~\cite{Sampson:2011PLDI}, is to 
impose equality constraints for \emph{all references}. This is akin to the bidirectional flow graph 
in~\secref{sec:CFL}. Our solution combines the ``bidirectional'' system  
with reference immutability to achieve limited subtyping and better precision.

One immediately notices the parallel with CFL-reachability. $\y <:  \x$ corresponds to forward edge 
$\y \arrow{d} \x$ in $G_{\sc RI}$ and the inverse constraint $\x <:  \y$ corresponds to the inverse edge $\x \arrow{d} \y$.
The same reasoning applies to all other explicit and implicit assignments: if the
left-hand-side is \readonly then the rule enforces a subtyping constraint, otherwise
it adds an inverse constraint, thus ensuing an equality just as in~\secref{sec:CFL}.

\begin{figure}[!t]
\begin{center}
\small
\begin{minipage}{11.35cm}
\begin{semantics}
\ntyperule{tassign}{
\Gamma \vdash \y <: \x \in \mathcal{C} \quad\quad \Gamma \vdash \x \mbox{ is not } {\readonly} \Rightarrow \x <: \y \in \mathcal{C} \quad\quad \Gamma \vdash \mathcal{C} \mbox{ holds } 
}{
\Gamma \vdash \x = \y
}
\end{semantics}
\end{minipage}

\begin{minipage}{12.5cm}
\begin{semantics}\ntyperule{twrite}{
    \Gamma \vdash \y <: \x \rhd \f \in \mathcal{C} \quad\quad
    \Gamma \vdash \x.\f \mbox{ is not } {\readonly} \Rightarrow \x \rhd \f <: \y \in \mathcal{C} \quad\quad \Gamma \vdash \mathcal{C} \mbox{ holds }
}{
  \Gamma \vdash \x.\f = \y
}
\end{semantics}
\end{minipage}

\begin{minipage}{12.15cm}
\begin{semantics}
  \ntyperule{tread}{
  \Gamma \vdash \y \rhd \f <: \x \in \mathcal{C} \quad\quad \Gamma \vdash \x \mbox{ is not } {\readonly} \Rightarrow \x <: \y \rhd \f \in \mathcal{C} \quad\quad \Gamma \vdash \mathcal{C} \mbox{ holds } 
}{
  \Gamma \vdash \x = \y.\f
}
\end{semantics}
\end{minipage}

\begin{minipage}{12cm}
\begin{semantics}\ntyperule{tcall}{
  \mathit{typeof}({\m}) = {\this}, {\p} \rightarrow {\ret} \\ 
\Gamma \vdash \y <: q_{\this}^i \rhd {\this} \in \mathcal{C} \quad\quad \Gamma \vdash \x \rhdri \this \mbox{ is not } {\readonly} \Rightarrow q_{\this}^i \rhd {\this} <: \y \in \mathcal{C} \\
\Gamma \vdash \z <: q_{\p}^i \rhd {\p} \in \mathcal{C} \quad\quad \Gamma \vdash \x \rhdri \p \mbox{ is not } {\readonly} \Rightarrow q_{\p}^i \rhd {\p} <: \z \in \mathcal{C} \\
\Gamma \vdash q_{\ret}^i \rhd {\ret} <: \x \in \mathcal{C} \quad\quad \Gamma \vdash \mbox{ is not } {\readonly} \Rightarrow  \x <: q_{\ret}^i \rhd {\ret} \in \mathcal{C} \\
\Gamma \vdash \mathcal{C} \mbox{ holds}
}{
  \Gamma \vdash i: \x = \y.\m(\z)
}
\end{semantics}
\end{minipage}
\end{center}
\caption{Typing rules associated to program statements.  $\Gamma \vdash A$ means $\mathcal{C},\sigma \vdash A$ as $\Gamma = \langle \mathcal{C}, \sigma\rangle$.
``$\mathcal{C}$ holds'' requires (1) that $\mathcal{C}$ is closed under 
the rules from~\figref{fig:transitive} and (2) that all constraints in $\mathcal{C}$ type check.}
\label{fig:types}
\end{figure}

Rules \rulename{twrite}, \rulename{tread} and \rulename{tcall} use adaptation.
At field accesses \y.\f, field \f is interpreted in the context of receiver \y. If $\f$ is \high
(or \low in the positive setting), then its adapted value remains \high
(or \low). If \f is \poly, then the adapted value assumes the type of $\y$. Notably, FlowCFL restricts
the type of \f to $\{ \high, \poly \}$ in the negative setting, and to $\{ \poly, \low \}$
in the positive setting. In our discussion going forward, we assume the negative setting (as described in~\secref{sec:overview}),
however all reasoning applies to the symmetric positive setting as well. 
We can allow \low fields in FlowCFL. However, we are interested in type inference, 
and allowing \low fields would create ambiguity: if $\x.\f$ flows to 
\low, (1) do we infer that field \f is \low, and is ``special'', i.e., it is excluded from the state of a potentially \high{} \x, 
or (2) do we infer that \f is just a ``regular'' field, and a negative $\x.\f$ entails a \low{} \x? 
Restricting fields to $\{ \high, \poly \}$ chooses the latter, as there is no way to know,
without programmer annotations, which fields are ``special''. Inference tools such as 
Javarifier~\cite{Quinonez:2008ECOOP} and ReImInfer~\cite{Huang:2012OOPSLA} make the same choice. 
The restriction states that a positive reference \x cannot have negative components. 

Rules \rulename{twrite} and \rulename{tread} handle heap-transmitted dependences. 
Consider 
\[\code{x.f = a;} \quad \code{y = x;} \quad \textsf{\textbf{neg}} \; \code{b = y.f}\]

Rules \rulename{twrite}, \rulename{assign}, and \rulename{tread} create constraints 
\[{\sf a} <: \x \rhd \f \quad \x <: \y \quad \y \rhd \f <: {\sf b} \]
Since ${\sf b} = \low$, $\f$ is \poly and we have
\[{\sf a} <: \x \quad \x <: \y \quad \y <: {\sf b} \]
which forces ${\sf a} = \low$, as needed. 
Notice again the parallel with CFL-reachability. Constraints 
\[{\sf a} <: \x \rhd \f \quad \x <: \y \quad \y \rhd \f <: {\sf b}\]
correspond to the path in $G_{\sc RI}$
\[ {\sf a} \arrow{$w_\f$} \x \arrow{d} \y \arrow{$r_\f$} {\sf \underline{b}} \]
and the ``linear'' constraints 
\[{\sf a} <: \x \quad \x <: \y \quad \y <: {\sf b}\] 
correspond to dropping the $w_\f$ and $r_\f$ annotations, thus achieving a CSFI approximation.
FlowCFL is a more precise variant of CSFI as we argue in~\secref{sec:equivalence}.

Rule \rulename{tcall} captures call-transmitted dependences and is the most
involved. Unlike previous systems, e.g., EnerJ and DroidInfer, FlowCFL allows for \emph{distinct} adapters. 
Every parameter/return has a distinct associated adapter $q_\this^i$, 
$q_{\p}^i$ and $q_{\ret}^i$ instead of a single per-call-site adapter $q^i$. This is
necessary to achieve the CFL-reachability semantics.
We elaborate on this shortly.

\begin{figure}[!t]
\begin{center}
\small
\begin{minipage}{8.25cm}
\begin{semantics}
    \ntyperule{erase-left}{
  \Gamma \vdash \x \rhd \f <: \y \in \mathcal{C} \quad\quad \Gamma \vdash \sigma(\f) = \poly
}{
  \Gamma \vdash \x <: \y \in \mathcal{C}
}
\end{semantics}
\end{minipage}

\begin{minipage}{8.25cm}
\begin{semantics}
    \ntyperule{erase-right}{
  \Gamma \vdash \y  <: \x \rhd \f \in \mathcal{C} \quad\quad \Gamma \vdash \sigma(\f) = \poly
}{
  \Gamma \vdash \y <: \x \in \mathcal{C}
}
\end{semantics}
\end{minipage}

\begin{minipage}{7.5cm}
\begin{semantics}
    \ntyperule{trans-local}{
  \Gamma \vdash \x  <: \y \in \mathcal{C} \quad\quad \Gamma \vdash \y  <: \z \in \mathcal{C}
}{
  \Gamma \vdash \x <: \z \in \mathcal{C}
}
\end{semantics}
\end{minipage}

\begin{minipage}{12.25cm}
\begin{semantics}
    \ntyperule{trans-call}{
  \Gamma \vdash \z  <: q^i_{\p} \rhd \p \in \mathcal{C} \quad\quad \Gamma \vdash q^i_{\ret} \rhd \ret  <: \x \in \mathcal{C} \quad \quad \Gamma \vdash \p <: \ret \in \mathcal{C}
}{
  \Gamma \vdash \z <: \x \in \mathcal{C}
}
\end{semantics}
\end{minipage}


\end{center}
\caption{Constraint propagation. \p and \ret in \rulename{trans-call} can be any of \this, \p, or \ret.} 
\label{fig:transitive}
\end{figure}

Before we delve into \rulename{tcall}, we consider the rules in~\figref{fig:transitive}.
The rules explicitly collect transitive intraprocedural constraints into $\mathcal{C}$; they 
capture constraints that correspond to call/ret balanced paths (i.e., $M$-paths). 
\rulename{erase-left} and \rulename{erase-right}
``linearize'' constraints---e.g., when $\f$ is \poly, 
$\y \rhd \f <: {\sf b}$ becomes $\y <: {\sf b}$. This corresponds to dropping field $w,r$ annotations, 
thus achieving a variant of CSFI. Notably, a
constraint is linearized only if the corresponding field is \poly, 
which happens only when the field is on a path to a sink. 

Rule \rulename{trans-call} in~\figref{fig:transitive} transfers
constraints from the callee to the caller. If there is flow from a parameter
\p to return \ret, captured by subtyping constraint $\p <: \ret \in \mathcal{C}$, then
there is flow from actual argument \z to the left-hand-side of the call assignment
\x. 

Consider the example below, which is similar to~\figref{fig:runningExample}:
\begin{center}
\begin{tabular}{ll}
\begin{minipage}{6.5cm}
\begin{lstlisting}
class A { 
  %\cpoly% B f;
  void set(%\cpoly% A this, %\cpoly% B p) {
    this.f = p;
  }
  %\cpoly% B get(%\cpoly% A this) {
    ret = this.f;
    return ret;
  }
}      
\end{lstlisting}
\end{minipage}
\begin{minipage}{6.5cm}
\begin{lstlisting} 
main() {
  %\clow% A e = new A; % \hspace{2mm}\small\it \fbox{o}%
  B a = new B;
  e.set(a); // %$q_{\this}^4 = q_{\p}^4 = \clow$%
  %\clow% A g = e;
  %\clow% B b = g.get(); // %$q_{\this}^6 = q_{\ret}^6 = \clow$%
}  
\end{lstlisting}
\end{minipage} 
\end{tabular}
\end{center}

Class {\sf A} is polymorphic and {\sf main} uses {\sf A} in a negative context.
As {\sf b} is a sink, it follows that {\sf a} flows to a sink and the types should properly reflect the flow.
Line 4 in {\sf A.set} results in constraint ${\sf p} <: \this \rhd \f$.
Since $\f$ is \poly, \rulename{erase-right} in~\figref{fig:transitive} produces
${\sf p} <: \this$. 
Call site 4 in {\sf main} entails 
\[{\sf a} <: q_\p^4 \rhd \p \quad {\sf e} <: q_\this^4 \rhd \this \quad q_\this^4 \rhd \this <: {\sf e}\] 
The last constraint is the inverse of the previous one due to the mutation of \this. 
Rule \rulename{trans-call} in~\figref{fig:transitive} combines constraints
\[{\sf a} <: q_\p^4 \rhd \p \quad \p <: \this \quad q_\this^4 \rhd \this <: {\sf e}\]
to get ${\sf a} <: {\sf e}$. 
Analogously, constraint $\this \rhd \f <: \ret$
in {\sf get} and call site 6 in {\sf main} yield ${\sf g} <: {\sf b}$. 
Constraints ${\sf a} <: {\sf e}$, ${\sf e} <: {\sf g}$ (due to \code{g = e}) and 
${\sf g} <: {\sf b}$ capture the flow from {\sf a} to {\sf b}.


\subsection{FlowCFL$^-$}

Why not use the following simpler \rulename{tcall}?
\begin{center}
\begin{minipage}{10.5cm}
\begin{semantics}\ntyperule{tcall}{
  \mathit{typeof}({\m}) = {\this}, {\p} \rightarrow {\ret}  \\ 
\Gamma \vdash \y <: q^i \rhd {\this} \quad\quad \Gamma \vdash \x \rhdri \this \mbox{ is not } {\readonly} \Rightarrow q^i \rhd {\this} <: \y  \\
\Gamma \vdash \z <: q^i \rhd {\p} \quad\quad \Gamma \vdash \x \rhdri \p \mbox{ is not } {\readonly} \Rightarrow q^i \rhd {\p} <: \z  \\
\Gamma \vdash q^i \rhd {\ret} <: \x  \quad\quad \Gamma \vdash \x \mbox{ is not } {\readonly} \Rightarrow \x <: q^i \rhd {\ret} 
}{
  \Gamma \vdash i: \x = \y.\m(\z)
}
\end{semantics}
\end{minipage}
\end{center}
The rule makes use of a single viewpoint adapter $q^i$ rendering $\mathcal{C}$ and the rules in~\figref{fig:transitive}
unnecessary! Replacing \rulename{tcall} in~\figref{fig:types} with the above \rulename{tcall} yields a new type system, 
which we call FlowCFL$^-$ (FlowCFL minus). An advantage of FlowCFL$^-$ is its simplicity; it is also sound, however, 
it rejects programs that CFL-reachability over $G_{\sc RI}$ handles precisely.

The following (somewhat contrived) example illustrates the imprecision of FlowCFL$^-$ and the need for multiple adapters:
\begin{center}
\begin{tabular}{ll}
\begin{minipage}{6.85cm}
\begin{lstlisting}
%\cpoly% Y m(%\cpoly% A this, %\cpoly% X p, %\cpoly% Y q) {
  this.f = p;
  ret = q;
}
...
A a, X x, Y y;
Y y2 = a.m(x,y);
%\clow% X x2 = a.f;

A a1, X x1, Y y1;
%\clow% Y y3 = a1.m(x1,y1);
\end{lstlisting}
\end{minipage}

& 

\begin{minipage}{6cm}
\begin{tikzcd}[cramped]
{\sf \underline{x2}}  & {\sf a} \arrow[l,"{r_\f}" description]  \arrow[dr,"{(_7}" description]  & & \y \arrow[dr,"{(_7}" description] & & \\
& {\sf a1} \arrow[r,"{(_{11}}" description]  & \this \arrow[ul, dashed, bend right, "{)_7}" description] \arrow[l, dashed, bend left, "{)_{11}}" description] & {\sf y1} \arrow[r,"{(_{11}}" description]  & {\sf q} \arrow[d] \\
& {\sf x} \arrow[r,"{(_7}" description] & \p \arrow[u, "{w_\f}" description] & {\sf y2}  &  \ret \arrow[l,"{)_7}" description] \arrow[dl,"{)_{11}}" description]  & \\
& {\sf x1} \arrow[ur,"{(_{11}}" description] & & {\sf \underline{y3}} & & 
\end{tikzcd}
\end{minipage}




\end{tabular}
\end{center}

There are two disjoint paths through \m: (1) ${\sf p} \patharrow{} \this$, 
and (2) ${\sf q} \patharrow{} \ret$. At call 7 the first path appears in negative context (as subpath of the path from \x to \low{} {\sf x2}), 
while the second path appears in positive context (as subpath of the path from {\sf y} to {\sf y2}). At call 11 the opposite happens:
the first path appears in positive context and the second one in negative context.
CFL-reachability precisely propagates the negative qualifier \low{} {\sf x2} back to {\sf x} and \low{} {\sf y3} back to {\sf y1}. 
FlowCFL propagates the negative qualifiers in exactly the 
same way. It discovers paths $\x \patharrow{} {\sf x2}$ and ${\sf y1} \patharrow{} {\sf y3}$ via $\mathcal{C}$: 
it adds $\x <: {\sf x2}$ and ${\sf y1} <: {\sf y3}$ to $\mathcal{C}$ based on $\p <: \this$ and ${\sf q} <: \ret$ respectively; it adds no
spurious constraints (i.e., paths).

In contrast, with a single adapter $q^i$
(e.g., as in DroidInfer and EnerJ) the above precise typing is impossible. 
This is because the role of the adapter is twofold: (1) to interpret the \poly parameter/return in the corresponding context, 
and (2) to propagate paths from callee to caller. 
Given sink \low{} {\sf X x2} in line 8, field \f is \poly (due to the flow of \f to sink {\sf x2}).
This forces \this and \p of \m to \poly, as shown in the typing of \m in lines 1-4. 
Sink \low{} {\sf Y y2} in line 11 forces \ret and {\sf q }to \poly as well. 
Due to $q^7 \rhd \this <: {\sf a}$ and $q^{11} \rhd \ret <: {\sf y3}$ respectively, 
we have $q^{7} = \low$ and $q^{11} = \low$. Thus, $\y <: q^7 \rhd {\sf q}$ and ${\sf a1} <: q^{11} \rhd {\sf this}$ unnecessarily 
force \y and {\sf a1} to \low.


Multiple adapters differentiate between flow paths. This is because the purpose of the adapters is 
only to interpret the \poly parameter/return in the corresponding context; propagation of
paths from calee to caller is done via $\mathcal{C}$. 
In our example we have ${\sf a1} <: q_\this^{11} \rhd \this$, ${\sf x1} <: q_\p^{11} \rhd \p$, ${\sf y1} <: q_{\sf q}^{11} \rhd {\sf q}$,
and $q_\ret^{11} \rhd \ret <: {\sf y3}$. $q_{\this}^{11} = q_{\p}^{11} = \high$, 
and $q_{\ret}^{11} = q_{\sf q}^{11} = \low$. The qualifiers are flipped at call site 7:
$q_{\this}^{7} = q_{\p}^{7} = \low$, and $q_{\ret}^{7} = q_{\sf q}^7 = \high$.

\section{Equivalence of CFL-reachability and Type-based Analyses}
\label{sec:equivalence}


Recall that CFL-reachability over both $w,r$ and call/ret annotations is undecidable. Typical approximations are CSFI (context-sensitive, field-insensitive)
and CIFS, and variants in-between. FlowCFL captures a variant of CSFI, which we call CSFI$^+$. 
As mentioned earlier, CSFI replaces all field annotations with {\sf d} and performs CR-reachability. E.g., in 
\[ {\sf x.f = a;} \;\; \textsf{\textbf{neg}} \; {\sf b = x.g; } \quad \Rightarrow \quad {\sf a} \arrow{d} \x \arrow{d} {\sf \underline{b}} \] 
CSFI replaces $w_\f$ and $r_\g$ with ${\sf d}$ and spuriously propagates \low ${\sf b}$
back to ${\sf a}$. Another way to look at CSFI is as if we replaced production $B \rightarrow w_\f \; B \; r_\f$ 
in the PG grammar in~\figref{fig:cfgs}(b) with $B \rightarrow w_\f \; B \; r_\g$.

Like CSFI, CSFI$^+$ does match certain distinct field annotations $w_\f$ and $r_\g$, but not all. 
CSFI+ matches $w_\f$ and $r_\g$ only if fields $\f$ and $\g$ both flow to sinks. 
As an example of potential imprecision in CSFI$^+$ consider the two snippets of the same program:

\begin{center}
\begin{tabular}{lllllll}
\begin{minipage}{2.5cm}
\begin{lstlisting}
x.f = a0;
x.g = b0;
%\clow% c0 = x.f;
d0 = x.g;
\end{lstlisting}
\end{minipage}
& 
$\Rightarrow$
&
\begin{tikzcd}[cramped]
{\sf a0} \arrow[r, "{w_\f}" description] & \x \arrow[r, "{r_\f}" description]  \arrow[dr, "{r_\g}" description]  & {\sf \underline{c0}} \\
{\sf b0} \arrow[ur, "{w_\g}" description] & & {\sf d0} \\
\end{tikzcd}
& 
\quad
& 
\begin{minipage}{2.6cm}
\begin{lstlisting}
y.f = a1;
y.g = b1;
c1 = y.f;
%\clow% d1 = y.g;
\end{lstlisting}
\end{minipage}
& 
$\Rightarrow$
&
\begin{tikzcd}[cramped]
{\sf a1} \arrow[r, "{w_\f}" description] & \y \arrow[r, "{r_\f}" description]  \arrow[dr, "{r_\g}" description]  & {\sf c1} \\
{\sf b1} \arrow[ur, "{w_\g}" description] & & {\sf \underline{d1}} \\
\end{tikzcd}

\end{tabular}
\end{center}
Since both $\f$ and $\g$ flow to sinks, CSFI$^+$ matches the distinct field
annotations. It propagates negative {\sf c0} to both {\sf a0} and {\sf b0}. 
Similarly, it propagates negative {\sf d1} to both {\sf a1} and {\sf b1}. 

\begin{figure}[!t]

\begin{tabular}{ll}
\begin{minipage}{0.55\linewidth}
\algrenewcommand\algorithmicindent{0.5em}

\small{
\begin{algorithmic}[1]
\Procedure{Cfl}{}
\State $P = \emptyset, F = \emptyset$
\State Add $n \patharrow{M} n$ to $P$ for all sinks $n$
\While {$P$ or $F$ changes}
\For {{\bf each} $s$ in Program} 
\State{{\sc {Edge}}{($\mathit{forward}(s)$)}} 
\State{{\sc {Edge}}{($\mathit{inverse}(s)$)}} 
\EndFor
\EndWhile
\EndProcedure
\end{algorithmic}

\vspace{0.1in}
\begin{algorithmic}[1]
 \Procedure{Edge}{$\x \arrow{t} \y$} // $\x \arrow{t} \y \in G_{\sc RI}$\\
 \hspace{0.05cm} {\bf if} $\y \in \{\ret,\this,\p\}$ {\bf then} Add $\y \patharrow{M} \y$ to $P$
\For{{\bf each} $\y \patharrow{N} n \in P$}
\State case {\sf t}, $N$ of 
\State \hspace{0.05cm} {\sf d}, \_ \; -> \; Add $\x \patharrow{N} n$, $\x \patharrow{M} \y$ to $P$ 
\State \hspace{0.05cm} $r_\f$, \_ \; -> \; Add $\x \patharrow{N} n$, $\x \patharrow{M} \y$ to $P$, Add \f to $F$
\State \hspace{0.05cm} $w_\f$, \_ \; -> \; {\bf if} {$\f \in F$} {\bf then} Add $\x \patharrow{N} n$, $\x \patharrow{M} \y$ to $P$
\State \hspace{0.05cm} $)_i$, \_ \; -> \; Add $\x \patharrow{R} n$ to $P$
\State \hspace{0.05cm} $(_i$, $M/C$ \; -> \; Add $\x \patharrow{C} n$ to $P$
\State \hspace{0.05cm} $(_i$, $R$ \; -> \; {\bf for each} {$\y \patharrow{M} \ret \retarrow{i} \z \patharrow{N'} n \in P$} {\bf do} 
\State \hspace{2.00cm} Add $\x \patharrow{N'} n$, $\x \patharrow{M} \z$ to $P$
\State \hspace{1.75cm} {\bf end for}
\EndFor
\State {\bf for each} $\x \patharrow{M} \y \in P$ {\bf do}  
\State \hspace{0.05cm} {\bf for each} $\y \patharrow{M} \ret \in P$ {\bf do}
\State \hspace{0.2cm} Add $\x \patharrow{M} \ret$ to $P$ 
\State \hspace{0.05cm} {\bf end for}
\State {\bf end for}

\EndProcedure
\end{algorithmic}
}
\end{minipage}

&

\begin{minipage}{0.5\linewidth}
\algrenewcommand\algorithmicindent{0.5em}

\begin{algorithmic}[1]
\Procedure{Types}{} 
\State $S(n) = \{ \low \}$ for all sinks $n$
\State $S(n) = \{ \low, \poly, \high \}$
\While {$S$ or $\mathcal{C}$ changes}
\For {{\bf each} $s$ in Program} 
\State{{\sc {Constraint}}($\mathit{forward}(s)$)}
\State{{\sc {Constraint}}($\mathit{inverse}(s)$)}
\EndFor
\EndWhile
\EndProcedure
\end{algorithmic}

\vspace{0.2in}
\begin{algorithmic}[1]
\Procedure{Constraint}{$c$} 
\State {\sc {Solve}}($c$)
\State case $c$ of
\State \hspace{0.05cm} $\x <: \y$ \; -> \; Add $\x  <: \y$ to $\mathcal{C}$
\State \hspace{0.05cm} $\x \rhd {\poly} <: \y$ \; -> \; Add $\x  <: \y$ to $\mathcal{C}$
\State \hspace{0.05cm} $\x <: \y \rhd \poly$ \; -> \; Add $\x  <: \y$ to $\mathcal{C}$
\State \hspace{0.05cm} $q_{\ret}^i \rhd \ret <: \x $ \; -> \; -
\State \hspace{0.05cm} $\x <: q_{\p}^i \rhd \p$ \; -> \;  
\State \hspace{0.2cm} {\bf for each} $\p <: \ret \in \mathcal{C}$,
\State \hspace{1.5cm} $q_{\ret}^i \rhd \ret <: \z \in \mathcal{C}$ {\bf do}
\State \hspace{0.5cm} Add $\x  <: \z$ to $\mathcal{C}$ 
\State \hspace{0.5cm} {\sc {Solve}}($\x  <: \z$)
\State \hspace{0.2cm} {\bf end for}
\State {\bf for each} $\x <: \y \in \mathcal{C}$ {\bf do}  
\State \hspace{0.05cm} {\bf for each} $\y <: \ret \in \mathcal{C}$ {\bf do}
\State \hspace{0.2cm} Add $\x  <: \ret$ to $\mathcal{C}$ 
\State \hspace{0.2cm} {\sc {Solve}}($\x  <: \ret$)
\State \hspace{0.05cm} {\bf end for}
\State {\bf end for}
\EndProcedure
\end{algorithmic}
\vspace{0.1in}

\end{minipage}

\end{tabular}

  \vspace{-0mm}
  \caption{
  Algorithms {\sc Cfl} and {\sc {Types}} assume a set of user-defined sinks. {\sc {Cfl}} computes $P$, which collects all paths $\x \patharrow{N} n$ from variables to sinks. 
  It iterates over program statements $s$ processing the edges $\x \arrow{t} \y \in G_{\sc RI}$ 
  and adding paths to $P$ by concatenating the ``terminal'' annotation {\texttt t} and 
  the ``nonterminal'' annotation $N$ according to the rules of the CR context-free grammar in~\figref{fig:cfgs}. 
  {\sc {Types}} initializes $S$, then iterates over program statements $s$ 
  removing qualifies from $S$ and collecting constraints in $\mathcal{C}$.  
  The algorithms elide details to highlight the ``parallel'' structure of the two systems. 
}
  \label{fig:algorithms}
\end{figure}

The problem is to find a set of paths that includes \emph{all} properly matched $w/r$ and call/ret paths to sinks in $G_{\sc RI}$.
Without loss of generality we assume that sinks are primitive types, i.e., the PG-component of every path from $v$ to a sink is either a $G$-path 
or a $B$-path.
CSFI$^+$ back-propagates sinks maintaining a set $F$ of fields that flow to sinks. The difference between precise 
propagation, CSFI$^+$, and CSFI lies in production $B \rightarrow w_\f \; B \; r_\g$. Precise propagation infers a 
$B$-path when \f = \g (as in~\figref{fig:cfgs}(b)), CSFI$^+$ infers a $B$-path only when $\{\f, \g\} \in F$, and CSFI infers a $B$-path in all cases.


\figref{fig:algorithms} presents two equivalent implementations of CSFI$^+$. Both algorithms implement FlowCFL in the negative setting---they 
start from a set of sinks and back-propagate those sinks via CSFI$^+$-reachability. Algorithm {\sc Cfl} collects all paths from variables to sinks
in $P$, as well as all balanced subpaths of these paths ($M$-paths).  
One can easily show (by induction on the length of the path) that {\sc Cfl} captures in $P$ all properly matched paths in $G_{\sc RI}$. 
Algorithm {\sc Types} makes use of the type system in~\secref{sec:typeBasedSemantics}. It computes a map $S$ from program 
variables to sets of qualifiers. $S$ is initialized as follows: $S(u) = \{ \low \}$ for each sink $u$, $S(\x) = \{ \high, \poly, \low \}$ 
for each variables \x, and $S(\f) = \{ \high, \poly \}$ for each field \f. 
{\sc Types} iterates through program statements; it infers new ``linear'' constraints and removes infeasible 
qualifiers from variable sets until it reaches a fixed point. Function {\sc Solve} takes a constraint, e.g., $\x <: \y$ and updates $S(\x)$. E.g., if $S(\y) = \{ \low \}$
and $S(\x) = \{ \high, \poly, \low \}$, {\sc Solve} removes \high{} and \poly{} from $S(\x)$ because neither is a subtype of \low.
As another example, consider constraint $\x \rhd \f <: \y$ where $S(\y) = \{ \poly, \low \}$, $S(\x) = \{ \high, \poly, \low \}$, and 
$S(\f) = \{ \high, \poly \}$. {\sc Solve} removes \high from both $S(\x)$ and $S(\f)$ because the constraint cannot be 
satisfied if either \x or \f is \high. Such fixpoint iteration has been used in previous work~\cite{Kiezun:2007ICSE,Tip:2011TOPLAS, Huang:2012ECOOP}.

{\sc Types} assigns sets of qualifiers to variables. To assign a final typing to a variable/field, we pick the \emph{maximal qualifier} according
to preference ranking \high > \poly > \low. One can see through case by case analysis that the \emph{maximal typing} type checks with the 
rules in~\figref{fig:types} and~\figref{fig:transitive}. Qualifiers $q^i_\this$, $q^i_\p$, $q^i_\ret$ can take any value that satisfies the maximal typing.

We argue correctness of our type-based analysis by establishing equivalence between {\sc Cfl} and {\sc Types}. 
\defref{def:soundness} states that if there is an $M/C$-path or an $R$-path from $\x$ to a sink $n$, then the maximal type of \x in $S$ 
is at least, respectively, \low{} or \poly. For example, if there is an $R$-path, the maximal typing is \poly or \low.

\begin{definition} \label{def:soundness} (Soundness) $P \Rightarrow S$ if and only if
\[
\begin{array}{llll} 
1. & \x \patharrow{M/C} n \in P & \Rightarrow & \mathit{max}(S(\x)) <: \low \\
2. & \x \patharrow{R} n \in P & \Rightarrow & \mathit{max}(S(\x)) <: \poly \\ 
\end{array}
\]
\end{definition}

\defref{def:precision} states that $\x$'s maximal type in $S$ implies a corresponding path in $P$. 
For example, maximal typing \poly means that there is a $R$-path but there is no $M/C$-path.

\begin{definition} \label{def:precision} (Precision) $S \Rightarrow P$ if and only if
\[
\begin{array}{llll} 
1. & \mathit{max}(S(\x)) = \low & \Rightarrow & \exists \; \x \patharrow{M/C} n \in P \\
2. & \mathit{max}(S(\x)) = \poly & \Rightarrow & \exists \; \x \patharrow{R} n \in P \; \wedge \; \nexists \; \x \patharrow{M/C} n \in P \\
3. & \mathit{max}(S(\x)) = \high & \Rightarrow & \mbox{no path from \x to any $n$ in $P$} \\
\end{array}
\]
\end{definition}

\begin{definition} \label{def:equivalence} (Equivalence) $P \simeq S$ 
if and only if $P \Rightarrow S$ and $S \Rightarrow P$.
\end{definition}


Let the Hoare triple denote parallel execution of {\sc {Edge}} and {\sc {Constraint}} on statement $s$:
\[\{ P, S \} \;\; \mbox{{\sc {Edge}}}(e(s)) \; || \; \mbox{{\sc {Constraint}}}(c(s)) \;\; \{ P', S' \} \]
The equivalence result comes from the following theorem:

\begin{theorem} If $\;\; P \simeq S \;\;$ and 
$\;\; \{P,S\}  \; \mbox{{\sc {Edge}}}(e(s)) \; || \; \mbox{{\sc {Constraint}}}(c(s)) \; \{ P', S' \} \;\;$ then $P' \simeq S'$.
\end{theorem}

\begin{proof} The proof is carried out by case-by-case analysis as in~\citep{Milanova:2018ECOOP}.\end{proof}



\section{Related Work}
\label{sec:relatedWork}

CFL-reachability dates decades back~\cite{Reps:1995POPL,Reps:2000TOPLAS}, yet it remains 
highly relevant. Zhang and Su~\cite{Zhang:2017POPL}, Spath et al.~\cite{Spath:2019POPL}, and 
Chatterjee et al.~\cite{Chatterjee:2018POPL}, among other works, present novel CFL-reachability 
approximations and algorithms with application to data dependence. Xu et al.~\cite{Xu:2009ECOOP}, 
and Lu and Xue~\cite{Lu:2019OOPSLA}, again among other works, present novel CFL-reachability-based 
points-to analyses. In all works, the concept of the inverse edge, introduced
by Sridharan et al.~\cite{Sridharan:2005OOPSLA}, factors in. Our work presents a formal treatment of 
the inverse edges and paths and a correctness argument for CFL-reachability over graphs with inverse edges.
Recent work by Li et al.~\cite{Li:2020PLDI}
presents a graph simplification algorithm for CFL-reachability that removes certain edges that do not contribute
to paths to sinks. This work nicely complements our work, as it can be applied on any CFL-reachability
graph, including $G_{\sc BI}$ and $G_{\sc RI}$; Li et al., demonstrate their technique using DroidInfer's 
graphs~\cite{Huang:2015ISSTA}, which are $G_{\sc RI}$ graphs. (We use DroidInfer's graphs in our experiments as well.) 
We have focused on understanding the 
dynamic semantics of flows, establishing soundness of the removal of certain inverse edges, and drawing 
a connection between CFL-reachability and type-based flow analysis. 

Type-based analysis has a long history as well~\cite{Palsberg:2001PASTE} and our analysis
falls into this line of work. Classical work on type-based taint (information flow) analysis includes 
work by Shankar et~al.~\cite{Shankar:2001USENIXSecurity}, Volpano et al.~\cite{Volpano:1996}, 
and Myers~\cite{Myers:1999POPL}. 

Few works have explored the connection between CFL-reachability
and type-based analysis. Milanova~\cite{Milanova:2018ECOOP} presents an interpretation of 
reference immutability in terms of CFL-reachability. We make use of this interpretation (\secref{sec:referenceImmutability}), 
however, we address a different and more difficult problem, as the nature of approximation in 
reference immutability~\cite{Tschantz:2005OOPSLA,Huang:2012OOPSLA,Milanova:2018ECOOP} renders inverse edges 
unnecessary and reachability analysis much simpler. 

Rehof and Fahndrich~\cite{Rehof:2001POPL} connect type-based flow analysis and 
CFL-reachability. 
However, Rehof and Fahndrich do not discuss  
mutable references and it is unclear how their analysis and interpretation, targeting a pure functional language, 
can handle mutable data and heap-transmitted dependences.  
On the other hand, Rehof and Fahndrich handle higher-order functions while we do not. 
An important direction of future work is extending our approach with handling of higher-order functions
which will enable application of the FlowCFL framework to the analysis of dynamic languages. 
Fahndrich et al.~\cite{Fahndrich:2000PLDI} apply the theory of~\cite{Rehof:2001POPL} 
to build a context-sensitive Steensgard-style points-to analysis for C, thus 
using equality constraints instead of subtyping constraints. As mentioned earlier, equality constraints is the standard 
approach to the handling of mutable references~\cite{Shankar:2001USENIXSecurity,Sampson:2011PLDI,Fuhrer:2005ECOOP}. 

\section{Conclusion}
\label{sec:conclusion}

We presented FlowCFL, a framework for type-based reachability analysis. 
We presented (1) a novel dynamic semantics, (2) correctness arguments for CFL-reachability 
over graphs with inverse edges, and (3) equivalence between a CFL-reachability analysis 
and a type-based reachability analysis.

%% file: appendix1.tex

\section{Applications of {FlowCFL}}
\label{app:applications}

\subsection{Approximate Computing: EnerJ and Rely}

In addition to taint analysis, another application domain of FlowCFL is approximate computing, which has received
significant attention~\cite{Sampson:2011PLDI, Carbin:2013OOPSLA, Bornholt:2014ASPLOS, Holt:2016SOCC}. 
Approximate computing relies on programming language technology such as type systems
and Hoare logic to reason about execution on unreliable hardware~\cite{Sampson:2011PLDI,Carbin:2013OOPSLA},
execution in the presence of probabilistic sensor data~\cite{Bornholt:2014ASPLOS}, and execution
on inconsistent and approximate cloud storage systems~\cite{Holt:2016SOCC}. An overarching 
issue is the separation of non-approximate and approximate parts of the program. 
Currently, all works require large number of manual annotations that explicitly separate the approximate 
variables and operations from the non-approximate ones.

\subsubsection{EnerJ}

\begin{figure}[t]
\begin{tabular}{ll}
  \begin{minipage}{7cm}
  \begin{lstlisting}
 public class IntPair {
   int x;
   int y;
   int numAdditions = 0;
   void addToBoth(IntPair this; int amount) {
      x += amount;
      y += amount;
      numAdditions++;
   }
 }   
\end{lstlisting}
  \end{minipage}

&
  
  \begin{minipage}{6cm}
    \begin{lstlisting}
 public class Example {
   public static void main() {
     IntPair i = new IntPair();
     i.addToBoth(10);
     ...
     IntPair j = new IntPair();
     j.addToBoth(k);
     @Precise z = j.x + j.y; 
  }
}  
    \end{lstlisting}
  \end{minipage}
\end{tabular}
  \caption{{\sf IntPair} from EnerJ~\cite{Sampson:2011PLDI}. Variable \z at line 8 in {\sf main} is 
  precise ({\sf @Precise} maps to \low in FlowCFL), and therefore, flow from approximate 
  data to \z is forbidden. FlowCFL infers that class {\sf IntPair} 
  is polymorphic: \x, \y, \this and {\sf amount} of {\sf addToBoth} are \poly (exactly as annotated in~\cite{Sampson:2011PLDI}).
  FlowCFL infers that {\sf i} in {\sf main} is {\sf @Approx}, while {\sf j} and {\sf k} are {\sf @Precise}; i.e.,
  it instantiates polymorphic {\sf IntPair} as {\sf @Approx} in the context of {\sf i}, and as {\sf @Precise}
  in the context of {\sf j}. Field {\sf numAdditions} in {\sf IntPair} is {\sf @Approx} because it does not flow to \z in either context
  (again, exactly as in~\cite{Sampson:2011PLDI}).
} 
\label{fig:EnerJ}
\end{figure}

EnerJ~\cite{Sampson:2011PLDI} partitions the program variables into {\sf @Approx} and {\sf @Precise} 
where {\sf @Approx} variables can be stored and used in energy-efficient storage. It requires non-interference for correctness:
an {\sf @Approx} variable cannot flow into a {\sf @Precise} one. The EnerJ type system 
can be cast as an instance of FlowCFL in the negative setting. {\sf @Approx} maps to \high, {\sf @Precise} maps
to \low, and {\sf @Context} maps to \poly. 
Programmers can annotate a set of {\sf @Precise} sinks designating values
that must be computed precisely. The system infers types for the rest of the variables maximizing
the approximate part of the program. \figref{fig:EnerJ} illustrates.


\subsubsection{Rely}

\begin{figure}[t]
\begin{center}
\begin{minipage}{11cm}
\begin{lstlisting}
class Newton {
  static float tolerance = 0.00001;
  static int maxsteps = 40;    
  static float F(float x) { ... }
  static float dF(float x) { ... }

  static float newton(urel float xs) { 
    float x, xprim;
    float t1, t2;
    int count = 0;
    
    x = xs;
    xprim = xs + 2*tolerance;
    while ((x - xprim >= tolerance) || (x - xprim <= -tolerance)) {
      xprim = x;
      t1 = F(x);
      t2 = dF(x);
      x = x - t1 / t2;
      if (count++ > maxsteps) break; 
    }
    if (!((x - xprim <= tolerance) && (x - xprim >= -tolerance))) {
      x = INFTY;
    } 
    return x;
  } 
}
\end{lstlisting}
\end{minipage}
\end{center}
  \caption{Newton's method from~\cite{Carbin:2013OOPSLA}. Input {\sf xs} in line 7 is annotated unreliable
  ({\sf urel} corresponds to \high in FlowCFL).
  FlowCFL fills in the remaining annotations. It infers that {\sf F} is {\sf poly float F(poly float x)} and so is {\sf dF}. Variables \x, {\sf xprim}, {\sf t1} and {\sf t2}
  are inferred {\sf urel} (as explicitly annotated in~\cite{Carbin:2013OOPSLA}).
  All operations, except for line 19, are unreliable (as in~\cite{Carbin:2013OOPSLA}).
} 
\label{fig:Rely}
\end{figure}

Another system in this domain, Rely~\cite{Carbin:2013OOPSLA}, reasons about execution on unreliable hardware. Again, 
programmers must explicitly annotate all unreliable variables (using the {\sf urel} annotation), as well as all operations on 
unreliable variables (e.g., unreliable + becomes +.). Unannotated variables and operations are considered reliable.
\cite{Carbin:2013OOPSLA} describe how Rely verifies a bound on the reliability of a computation with respect to the reliability of its input. 
For example, it verifies that the result of the computation in~\figref{fig:Rely} is at least $0.99*R({\sf xs})$, 
where $R({\sf xs})$ is the reliability of input {\sf xs}.  

Rely can be cast as an instance of FlowCFL in the positive setting. The {\sf urel} (unreliable) Rely annotation maps to \high, and the default reliable annotation maps 
to \low. Programmers annotate unreliable 
inputs with \high and FlowCFL infers types for the rest of the program, thus minimizing the unreliable partition. 
Unlike with EnerJ where we maximize the approximate partition and thus, energy savings, 
here we minimize the unreliable partition, which may improve on Rely's bound. (The
smaller the unreliable partition, the more precise the bound on the computation.) \figref{fig:Rely} illustrates
type inference for Rely. We have run all programs from~\cite{Carbin:2013OOPSLA,Carbin:2013TR} through FlowCFL 
and we have inferred the same types as annotated in~\cite{Carbin:2013OOPSLA,Carbin:2013TR}.

\begin{figure}[t]
  \begin{tabular}{lll}
  \begin{minipage}{5cm}
    \begin{lstlisting}
public class Data { 
  int d;
  void set(Data this, int p) { 
     this.d = p; 
  }
  int get(Data this) { 
     return this.d; 
  }   
}
  \end{lstlisting}
  \end{minipage}
  
  & 
  &
  
  \begin{minipage}{8cm}
    \begin{lstlisting}
public class Example {
  public void main() {
    Data ds = new Data(); 
    sensitive int s = ...; // sensitive source 
    ds.set(s); 
    int ss = ds.get();    
    Data dc = new Data(); 
    int c = ...;  
    dc.set(c); 
    int cc = dc.get();
  }
}
    \end{lstlisting}
  \end{minipage}
  \end{tabular}
  \caption{An example from JCrypt~\cite{Dong:2016PPPJ}.
  {\sf s} in line 4 of {\sf main} is a sensitive input (\high in FlowCFL), and all computation affected by {\sf s} must be secure. FlowCFL infers 
  that class {\sf Data} is polymorphic, and that {\sf ds} and {\sf ss} in {\sf main} are sensitive. The remaining 
  variables remain plaintext (\low in FlowCFL).
} 
\label{fig:JCrypt}
\end{figure}

\subsection{Secure Computation}

Yet another application of FlowCFL is secure computation. As clients 
increasingly outsource computation to untrusted cloud servers, there is pressing 
need to preserve confidentiality of data. This can be done through computation outsourcing~\cite{Shan:2018Surveys} or 
Multi-party Computation (MPC)~\cite{EvansMPCBook}. 
Unfortunately, secure computation is expensive. Fully homomorphic encryption~\cite{Cooney:2009,Gentry:2010,Gentry:2011}
is still prohibitively expensive. Partially homomorphic encryption, an essential building block in both computation outsourcing
and MPC, is still costly; for example, homomorphic addition over cypertexts is about 5X more expensive than addition
over plaintexts~\cite{Tetali:2015Thesis}.
Therefore, it is important to minimize portions of the program that require computation under secure computation 
protocols. 

JCrypt~\cite{Dong:2016PPPJ} is a system where programmers annotate a set of inputs as sensitive (i.e., \high{}) and JCrypt propagates 
those sensitive annotations throughout the program. For example, inputs from files in MapReduce applications, or secret-shared inputs 
in MPC will be annotated as sensitive. JCrypt is yet another instance of FlowCFL in the positive setting; it minimizes the
sensitive portion of the program, thus minimizing expensive secure computation. 
\figref{fig:JCrypt} illustrates JCrypt and inference with FlowCFL.